%% file: Main.tex
\documentclass[12pt, draftclsnofoot, onecolumn]{IEEEtran}

\input{PKG}

\allowdisplaybreaks

\makeatletter
\patchcmd{\@maketitle}
{\addvspace{0.5\baselineskip}\egroup}
{\addvspace{-.5\baselineskip}\egroup}
{}
{}

\usepackage[letterpaper,bindingoffset=0.2in,
left=.7in,right=.8in,top=.9in,bottom=1.2in,
footskip=.2in]{geometry}

\usepackage{xpatch}
\xapptocmd{\appendix}{%
}{}{\PatchFailed}

\begin{document}
	\title{{Identification for Inverse Gaussian Channels} \\[-.1em] \vspace{4mm}}
	\author{\vspace{2mm} \fontsize{13}{13} \selectfont Mohammad Javad Salariseddigh\IEEEauthorrefmark{1} and Christian Deppe\IEEEauthorrefmark{2}
		\vspace{8mm}
		\\
		\fontsize{12}{12} \selectfont \IEEEauthorrefmark{1}Technical University of Darmstadt\\
		\selectfont\IEEEauthorrefmark{2}Technical University of Braunschweig\vspace{9mm}
		\\
		\fontsize{12}{12} \selectfont \itshape
		\textcolor{verdigris!80!gray}{This work is offered in memory of Professor Rudolf Ahlswede (1938-2010), \\ a true pioneer, whose brilliant legacy illuminates the path of our research.}  \vspace{-0mm}
	}

	\maketitle

	\begin{abstract}
		We derive lower and upper bounds on the identification capacity of inverse Gaussian (IG) channels, a fundamental model for molecular communications in fluid environments. The analysis considers deterministic encoding schemes under a peak time constraint and characterizes the asymptotic growth of optimal codebook sizes. In the converse, by leveraging the Borel-Cantelli lemma and the super-exponential near-zero tail behavior of the IG distribution, we establish an almost-sure lower bound on the first-passage time of drifting diffusing molecules. This leads to the conclusion that the identification capacity exhibits super-exponential growth with the codeword length $n$, i.e., $\sim 2^{(n \log n)R_n}$, where $R_n$ is the coding rate.
	\end{abstract}

	\section{Introduction}
	
	In Shannon's classical message transmission problem \cite{S48,Shannon1957}, the aim is to reliably recover the transmitted message index $m$ from a set whose size grows exponentially with the codeword length $n$ as $2^{nR}$, where $R$ denotes the coding rate (bits/symbol). That is, decoding requires the receiver to perform a global hypothesis test over the entire message set. This forces the output to contain enough information to uniquely localize one codeword among all codewords, effectively extracting the full information content of the message. In contrast, in Ahlswede's identification problem \cite{AD89,Ahlswede06gen}, the receiver does not attempt to reconstruct $m$. Instead, for each candidate message $m'$, it performs an independent binary hypothesis test distinguishing between $\mathcal{H}_0: m = m'$ and $\mathcal{H}_1: m \neq m'$. The task is therefore not a single joint decision over all messages, but a collection of yes/no tests, each asking whether a given message is the true one, which enables vastly larger codebooks. Thus, while Shannon decoding forces a global resolution among all messages, identification only requires answering localized membership queries, analogous to checking whether a specific page is the target page without needing to read the entire book in detail to understand its content. In other words, the Shannon decoder must determine \emph{which} message occurred, whereas an identification test only asks \emph{whether} a specific message occurred.

	The identification problem has attracted significant attention, particularly in the realm of post-Shannon information theory and goal-oriented communication paradigms; see \cite{Goldreich12,Salariseddigh23_BSC_Future_Internet} for further background and applications. Notably, this framework has been considered for event detection in emerging 6G wireless systems \cite{boche2025novel,Salariseddigh23_BSC_Future_Internet} as well as in molecular communication (MC) scenarios \cite{Salariseddigh_PhD_Diss,akyildiz2015internet}. With randomized encoders \cite{AD89}, one observes that the number of identifiable messages exhibits doubly exponential growth in the codeword length, $n,$ whereas restricting to deterministic encoders \cite{AN99} typically leads to only exponential \cite{J85,Salariseddigh_IT} (finite input and output alphabets) or super-exponential growth in $n$ \cite{Salariseddigh_IT,Salari_ISI,Salari_Colored} (infinite input alphabet), reflecting a fundamental gap in achievable performance. Classical random coding arguments ensure the existence of identification codes with doubly exponential size \cite{AADT20,Sidorenko22}, however, these approaches rely on stochastic encoders or random number generators that are often impractical to implement and can be difficult to realize within MC systems. Identification problem for channels adapted specifically to MC scenarios including Poisson, generalized Poisson and binomial channels is studied in \cite{Salariseddigh-TMBMC,Salariseddigh_OJCOMS_23,Salariseddigh25_ITW,Salariseddigh_Binomial_ISIT}.
	Constructions of identification codes have been extensively studied in the literature, as discussed in \cite{KT99,Verdu02,Zinoghli24,Vorobyev25,Lengerke26}. A close interplay between \emph{identification and function computation} has been established in recent works \cite{Salariseddigh24,Rosenberger24,ZhuFrey26}. In this paper, we address the identification problem for MC systems modelled as inverse Gaussian (IG) channels.
	
	The IG channel has become one of the most established channel models for the rigorous analysis of MC systems operating in flow-assisted environments \cite[Ch. 8]{Nakano13}. In these systems, information is encoded in the emission timing of signaling molecules that propagate through a fluid medium toward a receiver. When a persistent drift component is present, arising for example in microfluidic channels or biological transport processes such as blood flow, the stochastic first-passage time of molecules is accurately described by the IG distribution \cite{Karatzas14}. Moreover, communication based on the time of release may be used in the brain at the synaptic cleft, where two chemical synapses communicate over a chemical channel \cite{Borst99}. This physical correspondence offers a principled and analytically grounded framework for characterizing molecular propagation. It thus advances the IG channel beyond the purely abstract treatment in \cite{Chhikara24}, establishing it as a model with direct applicability to realistic biochemical and nanonetwork scenarios \cite{Nakano13}. The IG distribution includes several distinctive features that are highly important for analysis:
	
	The IG distribution presents several challenges from an information-theoretic perspective. Being supported on $\mathbb{R}_{+}$, it is inherently asymmetric and exhibits a pronounced right tail, allowing rare but exceptionally large delay realizations to occur with non-negligible probability. Its tail probabilities decay more slowly than in Gaussian settings, so these extreme events can dominate error probabilities, weaken concentration effects, and reduce the sharpness of typicality- and large-deviation-based analyses. Moreover, the strong asymmetry making likelihood ratios and decision thresholds inherently one-sided and more difficult to calibrate. Furthermore, for a fixed shape parameter $\lambda$, its variance scales as $\mathrm{Var}(X)=\mu^{3}/\lambda$, where $\mu$ denotes the mean (signal level), implying that dispersion increases rapidly with the mean and complicating normalization and performance comparisons across operating regimes. The distribution is also highly sensitive to its parameters $(\mu,\lambda),$ i.e., drift or diffusion in the first-passage models, where small parameter variations can significantly reshape the tail behavior and skewness, making robust code/decoder design harder and asymptotic analysis more challenging.
	
	Unlike the Gaussian channel \cite{S48}, whose capacity has a closed-form expression in terms of the signal-to-noise ratio, the Shannon capacity of the IG channel under average- or peak-delay constraints remains unknown. Existing work has focused on capacity bounds, asymptotic analyses, and numerical evaluations. Capacity and mutual-information bounds under various delay constraints were derived in \cite{Li14,Eckford12,Srinivas12,Eckford07,Eckford08,Farsad18}, while extensions to multiple indistinguishable particles, finite particle lifetimes, and drift-based IG models were studied in \cite{Farsad16_2,Kadloor12}. In this paper, we investigate the identification problem for the IG channel with deterministic encoders under a peak-delay constraint. We establish an achievability result via classical sphere-packing arguments \cite{CHSN13} and derive a converse under continuity of the IG distribution. To the best of our knowledge, identification over the IG channel has not been studied previously.

	\subsection{Notations}
	
	Blackboard-bold symbols $\mathbbmss{X}, \mathbbmss{Y}, \mathbbmss{Z}, \ldots$ denote alphabets. Lowercase symbols $x, y, z, \ldots$ show constants or realizations of random variables (RVs); uppercase symbols $X, Y, Z, \ldots$ denote RVs. Bold lowercase letters $\mathbf{x}, \mathbf{y}, \mathbf{z}, \ldots$ denote row vectors. All logarithms are base $2$. Let $[\![M_n]\!] \triangleq \{1,\ldots,M_n\}$. The sets of nonnegative and real numbers are denoted by $\mathbb{R}_+$ and $\mathbb{R}$, respectively. For any positive integer $x$, $\Gamma(x)=(x-1)! \triangleq (x-1)(x-2)\cdots 1.$ We use the standard \emph{Bachmann--Landau} notation $o(\cdot)$, $\mathcal{O}(\cdot)$, $\omega(\cdot)$, and $\sim$. Here, $f(n)\sim g(n)$ denotes asymptotic equivalence, while $X\sim\mathcal{N}$ denotes that $X$ follows the distribution $\mathcal{N}.$ The Euclidean norm is denoted by $\|\mathbf{x}\|$. Let $\mathcal{S}_{\mathbf{x}_0}(n,r)=\{\mathbf{x}\in\mathbb{R}^n:\|\mathbf{x}-\mathbf{x}_0\|\le r\}$ denote the $n$-dimensional hypersphere of radius $r$ centered at $\mathbf{x}_0$, and $\mathbbmss{Q}_{\mathbf{0}}=\{\mathbf{x}\in\mathbb{R}^n:0\le x_t\le U,\ \forall t\in[\![n]\!]\}$ the $n$-dimensional hypercube of side length $U$. 
	Let $(\upOmega,\mathcal{F},\Pr)$ denote a probability space, where $\upOmega$ is the sample space, $\mathcal{F}$ is a $\sigma$-algebra of subsets of $\upOmega$, and $\Pr:\mathcal{F}\rightarrow[0,1]$ is a probability measure satisfying $\Pr(\upOmega)=1,$ and $\Pr\!\left(\bigcup_{n=1}^{\infty}E_n\right) = \sum_{n=1}^{\infty}\Pr(E_n),$ for every countable collection $\{E_n\}_{n=1}^{\infty}\subseteq\mathcal{F}$ of pairwise disjoint events. The elements of $\mathcal{F}$ are called \emph{events}. For every event $E\in\mathcal{F}$, $\Pr(E)$ denotes the probability of $E$. Let $(A_n)_{n\ge1}$ and $(B_n)_{n\ge1}$ be sequences of RVs on $(\upOmega,\mathcal{F},\Pr)$. We write $A_n \ge B_n$ eventually almost surely, if there exists an event $\Omega_0\in\mathcal{F}$ with $\Pr(\Omega_0)=1$ such that, for every $\omega\in\Omega_0$, there exists an integer $n(\omega)<\infty$ satisfying $A_n(\omega)\ge B_n(\omega), \forall\, n\ge n(\omega).$ The symbol \(\mathbin{\dot\cup}\) denotes a disjoint union operator. We denote IG channel by $\bIG$.

	\subsection{Organization}
	The remainder of this manuscript is structured as follows. Section~\ref{Sec.SysModel} outlines the essential preliminaries for MC systems characterized by the IG channels. The principal results for the IG channels are presented in Section~\ref{Sec.Res}. Finally, Section~\ref{Sec.Conclusion} concludes the paper and discusses potential directions for future investigations.

	\section{Problem Setup, Model, and Preliminaries}
	\label{Sec.SysModel}
	Here, we adopt the system model and establish preliminaries for identification coding and capacity.
	
	\subsection{Brownian Motion}

	Brownian motion refers to the irregular motion of pollen particles suspended in water, first observed by the botanist Robert Brown in 1828 \cite[Ch. 2]{Karatzas14}. Let $0 = t_0 < t_1 < \cdots < t_n < \infty$ denote a sequence of time instants, and let $\{B_t\}_{t \ge 0}$ be a continuous-time stochastic process representing the position at time $t$ of a particle undergoing Brownian motion and let us define the increments and time steps by $\Delta_{B}^{t_j} = (B_{t_{j+1}} - B_{t_j})$ and $\Delta_t^{j} = (t_{j+1} - t_j),$ respectively. Then, the increments $\{\Delta_{B}^{t_j}\}_{j=0}^{n-1}$ are independent and depend only on $\Delta_t^{j}.$ Moreover, we say that $B_t$ has drift $v$ if $\mathbb{E}[\Delta_{B}^{t_j}] = v\Delta_t^{j},$ and has variance parameter $\sigma^2$ if $\mathrm{Var}(\Delta_{B}^{t_j}) = \sigma^2 \Delta_t^{j}.$ Then, each increment is a Gaussian RV, i.e.,
	\begin{align}
		\Delta_{B}^{t_j} \sim \mathcal{N} \big( v\Delta_t^{j}, \sigma^2 \Delta_t^{j} \big).
	\end{align}
	Following the framework of \cite{Srinivas12}, we assume that the transmitter releases one or more molecules into the fluid medium at specified times, after which the molecules propagate toward the receiver through diffusion and advection (flow). The receiver observes the corresponding arrival times and uses them to infer the original transmission times. We further assume that molecules cannot pass beyond the receiver boundary; upon arrival, each molecule is absorbed and does not re-enter the medium. As a result, we model the propagation as a one-dimensional process. This formulation, however, extends to higher dimensions as well, provided the medium is isotropic, ensuring that diffusion is directionally uniform. More precisely, the IG first-passage model remains exact when the receiver geometry and transport dynamics reduce the molecular hitting event to the first crossing of a single one-dimensional coordinate governed by a constant-drift, constant-diffusion process with an absorbing boundary. An infinite planar absorbing receiver provides an important higher-dimensional example: although molecular motion may be two- or three-dimensional, the first-passage problem reduces to the coordinate normal to the plane, and the IG result therefore remains valid. This reduction is consistent with the standard treatment of diffusion-controlled first-passage and reaction problems in \cite[Ch.~6]{Nakano13} and \cite[Ch.~4]{Berg93}. However, this should not be interpreted as a universal higher-dimensional first-passage result. For finite, curved, or otherwise non-planar receivers—including spherical receivers, disks, finite patches, and cylinders—the hitting event generally depends on the full receiver geometry and cannot, in general, be reduced to a single one-dimensional coordinate crossing; consequently, the resulting first-passage-time distribution need not be IG.

	\subsection{Probability Density Function of Particle's Position}
	
	\begin{figure}[t]
		\begin{center}
		\input{Pos_PDF.tex}
		\end{center}
	\vspace{-4mm}
	\caption{Probability density $f_X(x;t)$ as a function of time $t$ for diffusion coefficient $D=10 \mu m^2/s$ fixed drift $v=10 \mu m/s$ and volatility $\sigma=2 \mu m$, shown for different time shots. Each curve represents the likelihood of the process being at a fixed position $x$ at time $t.$ As $x$ increases, the peak of the curve shifts to larger $t$, approximately where $vt \approx x$, reflecting the time required for the drift to reach that position.  Additionally, peaks become lower and broader for larger $x$ due to the increasing variance $\sigma^2 t$, which spreads the distribution over time.}
	\vspace{-5mm}
	\label{Fig.PDF_PP}
\end{figure}

Consider a fluid medium with positive drift velocity $v>0$ and free diffusion coefficient $D,$ where the Brownian process variance is given by $\sigma^2 = 2D.$ A molecule is released into this fluid at time $t = 0$ at position $x = 0.$ Under the Brownian process $\{ B_t \}_{t \ge 0},$ the probability density function (PDF) of the particle's position $x$ at time $t > 0$ is given by \cite{Karatzas14}
\begin{align}
	f_X(x;t) = \frac{1}{\sqrt{2\pi\sigma^2 t}} \exp \bigg( - \frac{(x - vt)^2}{2\sigma^2 t} \bigg) ,
\end{align}
see Appendix \ref{App.PDF_Pos} for detailed derivations. The trend of such distribution versus position in different time snapshots is represented in Figure \ref{Fig.PDF_PP}.
Since the receiver behaves as a perfectly absorbing boundary, we focus exclusively on the first hitting time $Z$ at this boundary. We assume that the transmitter is located at the origin, while the receiver lies along the axis of interest at a distance $d > 0$. Under this configuration, the first arrival time $Z$ is defined as the earliest time at which the particle reaches the point $d$, that is,
\begin{align}
	Z = \inf \{ t > 0 \,:\, X(t) = d \}.
\end{align}
Note that for the zero drift regime (pure diffusion), i.e., where $v=0,$ the PDF of $Z,$ denoted by $f_Z(z)$ reads \cite{Karatzas14,Chhikara24}
\begin{align}
	f_Z(z) = \frac{d}{\sigma\sqrt{2\pi}}\, z^{-3/2} \exp\!\left( -\frac{d^2}{2\sigma^2 z} \right), \quad z>0,
\end{align}
see Appendix \ref{App.Eq_IG_FPT}. If the drift velocity $v>0,$ the PDF of $Z,$ denoted by $f_Z(z)$ is the IG distribution \cite{Chhikara24}:
\begin{align*}
	f_{Z}(z) = 
	\begin{cases}
		\sqrt{\frac{\lambda}{2\pi}} z^{-3/2} \exp \Big( - \lambda (z - \mu)^2 / (2\mu^2 z ) \Big), & z > 0 ;
		\\
		\hspace{3cm} 0, & z \leq 0 ,
	\end{cases}
\end{align*}
where $\mu = d/v$ and $\lambda = d^2 / \sigma^2;$ Refer to Figure \ref{Fig.IG_Para_Eff} for the effect of the parameters $\mu$ and $\lambda$ on the IG PDF.
In Appendix \ref{App.Eq_IG_FPT}, we demonstrate how the first-passage time PDF can be expressed in the form of an IG PDF for appropriately chosen parameters. We denote such distribution by $Z \sim \text{IG}(\mu,\lambda).$ Following \cite{Srinivas12}, we observe that when $v = 0,$ the RV $Z$ no longer follows an IG distribution but instead follows a Lévy distribution; see \cite{Li14} for details. Moreover, for $v < 0,$ there exists a non-zero probability that the released molecule fails to reach the receiving boundary. Therefore, to ensure a well-posed framework throughout this paper, we assume that $v > 0.$

\subsection{Molecular Distinguishability and Observation Model}

\begin{figure}[t]
	\begin{center}
		\input{Inv_Gau_pdf.tex}
	\end{center}
	\vspace{-4mm}
	\caption{The blue green illustrates the effect of varying $\lambda$ with fixed $\mu=2$. As $\lambda$ increases, the distribution becomes more concentrated around the mean, resulting in a sharper and higher peak with reduced variance and lighter tails. The blue curves illustrate the effect of varying $\mu$ with fixed $\lambda=1$. Increasing $\mu$ shifts the distribution to the right while also increasing its spread, since the variance grows with $\mu^3/\lambda$. This leads to broader and more dispersed curves with heavier right tails. Overall, $\lambda$ primarily controls concentration (dispersion around the mean), while $\mu$ controls both the location and scale of the distribution.}
	\label{Fig.IG_Para_Eff}
\end{figure}
\begin{figure}[t]
	\begin{center}
		\input{Sys_Mdl.tex}
	\end{center}
	\caption{Brownian motion with drift toward an absorbing boundary. The turquoise-colored curve represents a single realization of a stochastic trajectory starting from a transmitter (blue-colored) at the origin. The particle motion consists of random diffusive fluctuations combined with a constant drift in the positive x-direction. The fluid diffusion coefficients $D=10\mu m^2/s$ and the motion consists of Gaussian-distributed random increments with variance $\sigma=2 \mu m$, combined with a constant drift velocity $v=0.001 mm/s$ (common range for capillary microfluidics) in the positive $x$-direction. The dashed vertical line at $x=d=0.15mm$ denotes an absorbing wall, at which the particle is removed from the system upon first contact (first-passage event). The yellow-colored marker indicates the final position of the particle at the end of the simulation.}
	\label{Fig.BM}
\end{figure}

To formulate the MC channel model, we assume that the continuous-time stochastic process $\{B_t\}_{t \ge 0}$ corresponding to distinct signaling molecules are mutually independent. To convey the information, the transmitter modulates the emission time $X$ of its molecule, that is, the information is encoded in the release time $X_i$ of molecule $i$. After release, molecule $i$ starts from the initial condition $B_{t_0}^{(i)}=0$ and propagates according to Brownian motion with drift velocity $v>0$ and variance parameter $\sigma^2$ until it reaches the absorbing receiver boundary. This Brownian motion continues until the molecule arrives at the receiver. See Figure \ref{Fig.BM}. Let us assume that the released molecule takes the random time $Z$ to travel to the receiver. Let $X_t \in \mathbb{R}_{+}$ and $Y_t \in \mathbb{R}_{+}$ denote the channel input and output RVs, representing the molecular release time at the transmitter and the corresponding arrival time at the receiver, respectively. We assume that the input signal (time symbol) experiences an additive independent and identically distributed (i.i.d.) IG noise with mean $\mu > 0$ and finite variance $\mu^3 / \lambda.$ That is, $Z_{t} \overset{\text{\tiny i.i.d.}}{\sim} \text{IG}(\mu,\lambda).$ 

We consider an unbounded diffusion environment in which molecular propagation is governed solely by Brownian motion and the absorbing receiver boundary, without additional obstructions, external forces, or environmental interactions. The receiver remains active indefinitely, so that every molecule that eventually reaches the absorbing boundary is detected. We further assume perfect synchronization between the transmitter and receiver, such that release times are known exactly and molecular arrival times can be measured without timing error. Under these conditions, for the $i$th molecule, let $X_i \in \mathbb{R}_{+}$ denote its release time, $Z_i \in \mathbb{R}_{+}$ its propagation delay (i.e., first-passage time), and $Y_i \in \mathbb{R}_{+}$ its arrival time at the receiver. The letter-wise channel input-output relation is given by
\begin{align}
	\label{Eq.Law_Letter}
	Y_i
	&= X_i + Z_i,
	\qquad i=1,\ldots,n,
\end{align}
Assuming that the propagation processes of distinct molecules are independent, the propagation delays are modeled as i.i.d. IG RVs:
\begin{align}
	Z_i
	&\overset{\mathrm{i.i.d.}}{\sim}
	\operatorname{IG}(\mu,\lambda),
	&
	\mathbb{E}[Z_t]
	&=\mu,
	&
	\operatorname{Var}(Z_i)
	&=\frac{\mu^3}{\lambda}.
\end{align}

A key observation-model assumption is required when multiple channel uses are considered. Following the distinction discussed in \cite{Srinivas12}, we assume that each signaling molecule carries a detectable identifier that enables the receiver to associate an observed arrival with its corresponding release event, or equivalently, with its channel-use index. Such an identifier may be realized, for example, through a distinguishable molecular type, chemical marker, or other physically detectable tag \cite{Farsad2016,Bicen2018}. The identifier serves solely as auxiliary side information for molecule-to-channel-use association: it is not part of the communication alphabet, does not carry information, and does not enter the channel input or codeword. We further assume that the identifier does not affect the propagation dynamics; hence, all signaling molecules experience the same Brownian propagation law and have the same first-passage-time distribution.

In the present model, the signaling molecules serve as the physical carriers of the transmitted symbols, while the information is encoded exclusively in their release times. Thus, one channel use corresponds to one signaling molecule, and a length-$n$ input strings is given by the vector of release times
\begin{align}
	\mathbf{X}
	&= (X_1,\ldots,X_n),
\end{align}
where $X_t$ denotes the release time of the molecule associated with channel use $t$. The corresponding molecular identifier enables the receiver to associate the detected arrival with the appropriate channel-use index, without introducing an additional information-bearing dimension to the channel.

This assumption is particularly important because the IG propagation time has unbounded support. Consequently, for any two release times $X_t<X_{t'}$, it is possible that the corresponding arrival times satisfy $Y_{t'}<Y_t$. Thus, chronological arrival order cannot, in general, be used to determine the channel-use index of an arriving molecule. Under the proposed molecular-distinguishability mechanism, however, the receiver observes the arrival time together with the corresponding molecular identifier and can associate each arrival with its respective channel-use index, regardless of the chronological ordering of arrivals. Hence, the channel-use index refers to the corresponding signaling molecule and not to the chronological order of molecular arrivals.

The molecular-distinguishability assumption does not impose any temporal separation between consecutive channel uses. In particular, we do not require the arrival of molecule $i$ to occur before the release of molecule $i+1$. Hence, no finite guard interval is imposed between consecutive channel uses. This is essential because the IG distribution has unbounded support, and therefore
\begin{align}
	\Pr\{Z_t>G\}
	&>0,
	\qquad \forall G<\infty.
\end{align}
Consequently, no finite guard interval can guarantee that all molecules associated with one channel use have arrived before the next channel use begins. A finite guard interval could only make the probability of temporal overlap small, rather than eliminate it completely.

Thus, the channel uses in the present model are not required to be temporally non-overlapping. Two or more signaling molecules associated with different channel uses may be simultaneously propagating through the medium, and their arrival times may occur in an arbitrary chronological order. The molecular identifiers allow the receiver to distinguish these arrivals and associate each one with its corresponding channel-use index. For example, even if
\begin{align}
	X_t < X_{t'},
	\qquad
	Y_{t'} < Y_t,
\end{align}
the receiver can still determine that $Y_t$ corresponds to channel use $t$ and $Y_{t'}$ corresponds to channel use $t'$. Therefore, temporal overlap and arrival reordering do not create ambiguity in the identity-indexed observation model.

For a message encoded into a length-$n$ input string, let $\fX,\fY,$ and $\fZ$ denote the vectors of release times, corresponding molecular arrival times, and propagation delays, respectively. Here, $X_t$ and $Y_t$ correspond to the same signaling molecule associated with channel use $t$. In particular, $Y_t$ denotes the arrival time of molecule $t$, rather than the $t$-th molecule to arrive chronologically. The channel can therefore be written compactly as
\begin{align}
	\label{Eq.Law_Vector}
	\mathbf{Y}
	&= \mathbf{X} + \mathbf{Z}.
\end{align}
where $\fX, \fY,$ and $\fZ$ are input, output, and noise vector, respectively, i.e., $\fX = (X_t)_{t=1}^n,\fY = (Y_t)_{t=1}^{n},$ and $\fZ = (Z_t)_{t=1}^{n}.$ Under the i.i.d. IG noise assumption, the joint density of $\mathbf{Z}\in\mathbb{R}_{+}^{n}$ is
\begin{align}
	\label{Eq.Noise_PDF_Vect}
	f_{\mathbf{Z}}(\mathbf{z})
	&=
	\left(\frac{\lambda}{2\pi}\right)^{\frac{n}{2}}
	\prod_{t=1}^{n} z_t^{-\frac{3}{2}}
	\exp\left(
	-\frac{\lambda}{2\mu^2}
	\sum_{t=1}^{n}
	\frac{(z_t-\mu)^2}{z_t}
	\right).
\end{align}
The release time of each molecule is constrained according to
\begin{align}
	0 \leq X_t \leq T_{\max}, \qquad t \in [\![n]\!] ,
\end{align}
where $T_{\max}>0$ denotes the maximum allowable release time per channel use. Importantly, \eqref{Eq.Law_Vector} does not imply that the physical arrival times are ordered according to their channel-use indices. The physical arrival process may instead be represented by a permutation $\pi$ such that
\begin{align}
	Y_{\pi(1)}
	&<Y_{\pi(2)}
	<\cdots
	<Y_{\pi(n)},
\end{align}
where $\pi$ depends on the random propagation delays. The receiver uses the molecular identifiers to associate each chronological arrival with its corresponding channel-use index and thereby constructs the ordered vector $(Y_1,\ldots,Y_n)$. Thus, the component-wise additive channel law in \eqref{Eq.Law_Vector} remains valid despite possible temporal overlap and arrival reordering.

The above observation model therefore assumes neither bounded propagation delays nor strictly non-overlapping channel uses. Instead, it permits arbitrary temporal overlap and arrival reordering while relying on physically distinguishable molecules to preserve the correspondence between each release time and its associated arrival time. Under this explicit molecular-distinguishability assumption, one channel use corresponds to one distinguishable molecule, and the resulting identity-indexed additive IG channel model is well defined.

\subsection{Identification Coding}
Here, we present the code and capacity definitions for $\bIG.$

\begin{definition}[Identification Code]
	\label{Def.Inverse_Gaussian_Code}
	
	Let $0<e_1,e_2<1$ be prescribed type I and type II error bounds, and let $M_n\in\mathbb N$ denote the number of messages for a codeword length $n$. An $(n,\allowbreak M_n, \allowbreak e_1, \allowbreak e_2)$-code for $\bIG$ consists of a codebook $\mathbbmss{C} = \{ \fc_i \}_{i=1}^{M_n}$ where $\fc_i = ( c_{i,t} )_{t=1}^n$ fulfill $0 \leq c_{i,t} \leq T_{\rm max}$ and a collection of decoding regions $\D = \{ \mathbbmss{D}_i \}_{i=1}^{M_n}.$ Given a message $i \in [\![M_n]\!]$, the encoder sends $\mathbf{c}_i$, and the decoder's task is to address a binary hypothesis: Was a target message $j,$ sent or not? The type I and type II errors are characterized by:
	\begin{align}
		P_{e,1}(i) & = 1 - \int_{\mathbbmss{D}_i} \hspace{-2mm} f_{\fZ}(\fy - \fc_i) \, d\fy \quad \forall i \in [\![M_n]\!] ,
		\\
		P_{e,2}(i,j) & = \int_{\mathbbmss{D}_j} \hspace{-2mm} f_{\fZ}(\fy - \fc_i) \, d\fy \quad \forall i,j \in [\![M_n]\!],\, i \neq j .
	\end{align}
	The maximal type I and type II error probabilities read
	\begin{equation*}
		\bar{P}_{e,1}
		\triangleq
		\max_{i}P_{e,1}(i), \quad \text{and} \quad \bar{P}_{e,2}
		\triangleq
		\max_{ i\neq j}
		P_{e,2}(i,j) ,
	\end{equation*}
	where it holds $\bar{P}_{e,1} \leq e_1$ and $\bar{P}_{e,2} \leq e_2 .$
	\qed
\end{definition}

The codebook size in the identification setting considered here grows on a
super-exponential scale, with
$\log M_n=\Theta(n\log n),$
rather than on the classical linear-in-$n$ scale. In particular, although
the classical Shannon coding rate is given by $(\log M_n)/n,$
it grows as $\Theta(\log n)$ in the present setting and therefore does not
yield a finite asymptotic rate. Hence, the relevant asymptotic scale for
measuring the coding rate is $n\log n$. This motivates the following
definition.

\begin{definition}[Identification Coding Rate]
	\label{Def.Coding_Rate}
	For an $(n,M_n,e_1,e_2)$-code, the normalized identification coding rate
	is defined as
	\begin{align}
		\label{Eq.Normalized_Coding_Rate}
		R_n
		\triangleq
		\frac{\log M_n}{n\log n}.
	\end{align}
	This normalization yields a finite rate on the relevant asymptotic scale. \qed
\end{definition}

\begin{definition}[Identification Capacity] A rate $R>0$ is called achievable identification-rate if for every $0 < e_1,e_2 < 1$ there exists a sequence of $(n,\allowbreak M_n, \allowbreak e_1, \allowbreak e_2)$-codes $(\mathbbmss{C}^{(n)},\D^{(n)})_{n \in \mathbb{N}}$ such that the corresponding coding rates $(R_n)_{n\in\mathbb{N}}$ satisfy $\liminf_{n\to\infty} R_n \ge R.$ The identification capacity of $\bIG$ is the supremum of all achievable rates, i.e., $$\mathbb{C}_{\mathrm I}(\bIG) \triangleq \sup\left\{ R>0: \text{$R$ is an achievable rate} \right\}.$$
	\qed
\end{definition}

\subsection{Target-Rate Parametrization}

For a prescribed normalized identification rate $R$, in analogy with
the conventional parametrization $M_n=2^{nR}$ in Shannon's setting,
we may parameterize the codebook size on the present
$n\log n$ scale as
\begin{align}
	\label{Eq.Codebook_Size_Target_Rate}
	M_n(R)
	\triangleq
	2^{n\log n\,R}.
\end{align}
With this parametrization, the corresponding normalized coding rate
satisfies
\begin{align}
	\frac{\log M_n(R)}{n\log n}
	=R.
\end{align}
Thus, $M_n(R)$ represents the codebook size associated with a
prescribed target rate $R$.

It is important to distinguish this target-rate parametrization from
the construction used in our achievability proof below. In the latter,
the codebook is constructed directly, yielding a sequence of codebook
sizes $(M_n)_{n\in\mathbb N}$. The resulting normalized coding rates
are then determined by
\begin{align}
	R_n
	\triangleq
	\frac{\log M_n}{n\log n},
\end{align}
and the achievable rate is established by showing that
\begin{align}
	\liminf_{n\to\infty} R_n \geq R.
\end{align}
Hence, the present proof follows the construction-first approach
\[
M_n
\;\longrightarrow\;
R_n
\;\longrightarrow\;
R,
\]
rather than fixing a target rate $R$ in advance and constructing a
code with exactly $M_n(R)$ codewords. The two viewpoints are
equivalent in spirit, but the latter is useful when one wishes to
formulate achievability explicitly for a prescribed target rate.

\section{Identification Capacity of IG Channel}
\label{Sec.Res}
In this section we give our capacity theorem and its proof.

\subsection{Main Results}
\label{Subsec.Main_Results}

Our main contribution is stated in the following theorem.

\begin{theorem}
	\label{Th.IG_Capacity}
	Consider the IG channel $\bIG.$ Under a peak-amplitude constraint and a
	super-exponential codebook-size regime with $M_n(R)=2^{(n\log n)R}$, the
	identification capacity of \(\bIG\) satisfies
	\begin{align}
		\label{Ineq.LU}
		\frac{1}{4}
		\leq
		\mathbb{C}_{\rm I}(\bIG)
		\leq
		\frac{3}{2}.
	\end{align}
\end{theorem}

\begin{proof}
	The proof combines the achievability and converse arguments established in Propositions~\ref{Propos.Achievability} and
	\ref{Propos.Converse}, respectively.
	
	For the lower bound, the achievability proof of
	Proposition~\ref{Propos.Achievability}, together with
	\eqref{Ineq.R_n_LB}, establishes that, for every $0<b<1$, the
	corresponding sequence of codes satisfies
	\begin{align}
		\liminf_{n\to\infty} R_n
		\geq
		\frac{1-b}{4}.
	\end{align}
	Thus, every rate $R<\frac14$
	is achievable. Equivalently, the set of achievable identification
	rates contains
	\begin{align}
		\label{Eq.Achiev_Rates}
		\left\{
		\frac{1-b}{4}:0<b<1
		\right\}
		=
		\left(0,\frac14\right).
	\end{align}
	Consequently, the supremum of the set of achievable identification rates $R,$ i.e., the identification capacity, is greater or equal than the supremum of set defined in \eqref{Eq.Achiev_Rates}. That is, since $b$ can be chosen arbitrarily small, we have
	\begin{align}
		\mathbb{C}_{\rm I}(\bIG) \geq \sup_{0<b<1} \left\{ \frac{1-b}{4} \right\} = \frac{1}{4} .
	\end{align}
	
	For the upper bound, let $R>0$ be an achievable identification rate. By definition, for
	every pair of error bounds $0<e_1,e_2<1$, there exists a sequence
	of $(n,M_n,e_1,e_2)$-codes whose normalized coding rates satisfy
	\begin{align}
		\liminf_{n\to\infty}R_n\geq R.
	\end{align}
	For this sequence, the converse bound of Proposition \ref{Propos.Converse} yields
	\begin{align}
		R
		\leq
		\liminf_{n\to\infty}R_n
		\leq
		\limsup_{n\to\infty}R_n
		\leq
		\frac32.
	\end{align}
	Thus, every achievable rate satisfies \(R\le 3/2\). Taking the supremum over all achievable identification-rates, we conclude that
	\begin{equation}
		\mathbb{C}_{\rm I}(\bIG)
		=
		\sup\big\{R:\,R\text{ is achievable}\big\}
		\le
		\frac32.
	\end{equation}
\end{proof}

\begin{proposition}[Achievability]
	\label{Propos.Achievability}
	Under the assumptions of Theorem~\ref{Th.IG_Capacity}, for every
	$b\in(0,1)$, there exists a sequence of identification codes
	satisfying
	\[
	\liminf_{n\to\infty} R_n \geq \frac{1-b}{4}.
	\]
	Consequently, every rate $R<\frac14$ is achievable.
\end{proposition}

\begin{proof}
	The proof is given in Subsection~\ref{Subsec.Achievability}.
\end{proof}

\begin{proposition}[Converse]
	\label{Propos.Converse}
	Under the assumptions of Theorem~\ref{Th.IG_Capacity}, for every
	$b\in(0,1)$, any admissible sequence of identification codes
	satisfies
	\[
	\limsup_{n\to\infty} R_n
	\leq \frac{3}{2}+2b.
	\]
	Consequently, every rate $R>\frac{3}{2}$ is unachievable.
\end{proposition}

\begin{proof}
	The proof is given in Subsection~\ref{Subsec.Converse}.
\end{proof}

\subsection{Achievability (Lower Bound)}
\label{Subsec.Achievability}

We construct a codebook via dense packing of hyperspheres in the eligible input space induced by the peak time constraint, exploiting its geometric structure for efficient encoding. We then define an appropriate decoding rule (distance decoder) and establish achievability by showing that both type I and type II error probabilities vanish asymptotically as $n \to \infty.$

\textbf{\textcolor{mycolor12}{Codebook Construction.}}
In the following, we construct the feasible codebook denoted by
\begin{align*}
	\mathbbmss{C} = \big\{ \fc_i \in \mathbb{R}^n:\; 0 \leq c_{i,t} \leq T_{\rm max} , \forall \, i \in [\![M_n]\!] , \forall \, t \in [\![n]\!] \big\} ,
\end{align*}
which is a subset of the hypercube $\mathbbmss{Q}_{\f0}(n,T_{\rm max})$ where $\fc_i = (c_{i,1} , \cdots, c_{i,n}).$ We use a packing arrangement of pairwise disjoint hyperspheres in the hypercube $\mathbbmss{Q}_{\f0}(n,T_{\rm max}).$ Let $\mathscr{S} = \{\S_{\mathbf{c}_i}(n,r_0)\}_{i=1}^{M}$ denote such a sphere packing consisting $M$ hyperspheres of radius $r_0 = \sqrt{n\epsilon_n}$ centered at codewords $\mathbf{c}_i$ inside $\mathbbmss{Q}_{\f0}(n,T_{\rm max})$ where
\begin{equation}
	\epsilon_n = \frac{a}{n^{( 1 - b) / 2}},
\end{equation}
with $a,b> 0$ being fixed and arbitrarily small constants. Accordingly, in the codebook construction, we impose the following requirements:
\begin{itemize}[leftmargin=*]
	\item In contrast to classical sphere packing formulations \cite{CHSN13}, where each sphere is required to be fully contained within the ambient constraint set, we adopt a relaxed boundary model, i.e., it is sufficient that centers $\mathbf{c}_i$ lie in $\mathbbmss{Q}_{\mathbf{0}}(n,T_{\max}),$ i.e., satisfying the constraint $0 \leq c_{i,t} \leq T_{\rm max}.$
	\item Hyperspheres are pairwise disjoint, allowing a clean rate characterization via standard volume arguments, i.e., ratio of the constraint-space volume to that of a single sphere.
	\item Each sphere has a non-empty intersection with $\mathbbmss{Q}_{\mathbf{0}}(n,T_{\max})$
\end{itemize}

Under this model, the packing density $\Updelta_n(\mathscr{S})$ of a finite sphere packing in the constrained hypercube is defined as the ratio of the total volume occupied by the spheres of a saturated packing within the constraint set to the volume of the constraint set, i.e.,
\begin{align}
	\Updelta_n(\mathscr{S}) \triangleq \frac{\text{Vol}\Big(\mathbbmss{Q}_{\f0}(n,T_{\rm max}) \cap \bigcup_{i=1}^{M_n}\S_{\fc_i}(n,r_0) \Big)}{\text{Vol}\big(\mathbbmss{Q}_{\f0}(n,T_{\rm max}) \big)}
	\label{Eq.Def_Density}
\end{align}
see \cite{CHSN13} for more discussions on the density's properties. A saturated packing argument is employed here, in the spirit of the Minkowski--Hlawka theorem for analyzing saturated sphere packings, as discussed in \cite{CHSN13}. This quantity measures the fraction of the constrained hypercube
occupied by the finite packing. However, no asymptotic lower bound on
$\Updelta_n(\mathscr{S})$ is required in the argument below.

Indeed, a packing-density result for spheres in the whole Euclidean
space does not \emph{automatically} imply a corresponding lower bound on the
density of a finite packing whose centers are constrained to a bounded
hypercube. In particular, boundary truncation affects the volume of the
packing spheres that lies inside the constraint set. Thus, a
whole-space packing-density result cannot, without an additional
argument, be regarded as a valid density bound for the present
finite-blocklength construction. We therefore avoid invoking any bound on $\Updelta_n(\mathscr{S})$ and
instead derive the required lower bound directly from the saturation
property of the packing. The resulting argument depends only on the
geometry of the constrained set, the non-overlap condition, and an
elementary volume-covering argument.

\textbf{\textcolor{mycolor12}{Rate Analysis.}}
Concretely, we consider a saturated packing $\bigcup_{i=1}^{M_n} \mathcal{S}_{\mathbf{c}_i}(n,r_0)$ of Euclidean spheres with radius $r_0 = \sqrt{n\epsilon_n},$ constrained to lie within the hypercube $\mathbbmss{Q}_{\mathbf{0}}(n,T_{\max}),$ i.e., whose centers satisfy
\begin{equation}
	\mathbf c_i
	\in
	\mathbbmss{Q}_{\mathbf 0}(n,T_{\max}),
	\qquad
	i\in[\![M_n]\!].
	\label{Eq:Centers_In_Cube}
\end{equation}
Since every packing sphere has radius $r_0$, the non-overlap condition
requires the distance between every pair of distinct packing centers to
be at least twice the sphere radius. Hence,
\begin{equation}
	\|\mathbf c_i-\mathbf c_j\|_2
	\geq
	2r_0,
	\qquad
	i\neq j.
	\label{Eq:Center_Separation}
\end{equation}
Thus, the packing centers form a $2r_0$-separated set in
$\mathbb{R}^n$.

The key consequence of saturation is that this separation property
induces a covering property. Suppose, for contradiction, that there
exists a point $\mathbf x \in \mathbbmss{Q}_{\mathbf 0}(n,T_{\max})$
such that
\begin{equation}
	\|\mathbf x-\mathbf c_i\|_2
	>
	2r_0,
	\qquad
	\forall i\in[\![M_n]\!].
	\label{Eq:Contradiction_Point}
\end{equation}
Then a further sphere $\S_{\mathbf x}(n,r_0)$
could be added to the packing. Indeed, its center $\mathbf x$ belongs
to the admissible center region, while
\eqref{Eq:Contradiction_Point} guarantees that its center is more than
$2r_0$ away from every existing packing center. Consequently,
$\S_{\mathbf x}(n,r_0)$ is non-overlapping with every sphere in
$\mathscr{S}$, contradicting the saturation of the packing. Therefore,
\begin{equation}
	\forall\mathbf x
	\in
	\mathbbmss{Q}_{\mathbf 0}(n,T_{\max}),
	\qquad
	\min_{i\in[\![M_n]\!]}
	\|\mathbf x-\mathbf c_i\|_2
	\leq
	2r_0.
	\label{Eq:Saturation_Covering}
\end{equation}
Equivalently, the packing centers constitute a $2r_0$-covering of the
entire admissible center region:
\begin{equation}
	\mathbbmss{Q}_{\mathbf 0}(n,T_{\max})
	\subseteq
	\bigcup_{i=1}^{M_n}
	\S_{\mathbf c_i}(n,2r_0).
	\label{Eq:Cube_Covering}
\end{equation}
Notice that the spheres in the original packing have radius $r_0$,
whereas the spheres in \eqref{Eq:Cube_Covering} have radius $2r_0$.
The latter are covering spheres introduced solely for the purpose of
the volume argument and are not additional packing spheres.

Taking the $n$-dimensional volume of both sides of \eqref{Eq:Cube_Covering} and applying the subadditivity of volume yields the following standard volume-covering bound:

\begin{align}
	\operatorname{Vol}
	\Big(
	\mathbbmss{Q}_{\mathbf 0}(n,T_{\max})
	\Big)
	\leq
	\operatorname{Vol}
	\left(
	\bigcup_{i=1}^{M_n}
	\S_{\mathbf c_i}(n,2r_0)
	\right)
	\leq
	\sum_{i=1}^{M_n}
	\operatorname{Vol}
	\Big(
	\S_{\mathbf c_i}(n,2r_0)
	\Big)
	=
	M_n
	\operatorname{Vol}
	\Big(
	\S_{\fc_1}(n,2r_0)
	\Big).
	\label{Eq:Volume_Covering}
\end{align}
Consequently, the total number of codewords (i.e., packed spheres) or achievable codebook size $M_n$ is bounded by
\begin{equation}
	M_n
	\geq
	\frac{
		\operatorname{Vol}
		\Big(
		\mathbbmss{Q}_{\mathbf 0}(n,T_{\max})
		\Big)
	}{
		\operatorname{Vol}
		\Big(
		\S_{\fc_1}(n,2r_0)
		\Big)
	}.
	\label{Eq:M_Lower_Volume}
\end{equation}

To further simplify \eqref{Eq:M_Lower_Volume}, we invoke Lemma~\ref{Lem.Neg_Log_Vol}, which provides a lower bound on the negative logarithm of the volume of an \(n\)-dimensional hypersphere, and exploit the volume formulas for the \(n\)-dimensional hypercube and hypersphere.  

\begin{customlemma}{1}[Log-Volume of Hypersphere]
	\label{Lem.Neg_Log_Vol}
	Let $n \geq 1$ and $r_0 >0.$ Then, the negative log-volume of $\S_{\fc_1}(n,r_0)$ reads
	\begin{align}
		- \log \text{Vol}(\S_{\fc_1}(n,r_0))
		= \frac n2\log\frac n2
		-n\log r_0
		-\frac n2\log(\pi e)
		+o(n).
	\end{align}
\end{customlemma}
\begin{proof}
	The proof is provided in Appendix \ref{App.Neg_Log_Vol}.
\end{proof}

Next, obeserve that the volume of the $n$-dimensional hypercube
$\mathbbmss{Q}_{\mathbf 0}(n,T_{\max})$ is given by
\begin{equation}
	\operatorname{Vol}
	\Big(
	\mathbbmss{Q}_{\mathbf 0}(n,T_{\max})
	\Big)
	=
	T_{\max}^{n}.
	\label{Eq:Cube_Volume}
\end{equation}
Moreover, the volume of an $n$-dimensional Euclidean hypersphere of radius $r$ is
\begin{equation}
	\operatorname{Vol}
	\Big(
	\S_{\fc_1}(n,r)
	\Big)
	=
	\frac{\pi^{n/2}}
	{\Gamma\!\left(\frac{n}{2}+1\right)}
	r^n.
	\label{Eq:Sphere_Volume}
\end{equation}

Thereby, \eqref{Eq:M_Lower_Volume} can be written as follows
\begin{align}
	\log M_n
	& \geq
	n\log T_{\max}
	-
	\log
	\operatorname{Vol}
	\Big(
	\S_{\fc_1}(n,2r_0)
	\Big)
	\nonumber\\
	&\geq
	n\log T_{\max}
	+
	\frac{n}{2}\log\frac{n}{2}
	-
	n\log(2r_0)
	-
	\frac{n}{2}\log(\pi e)
	+
	o(n)
	\nonumber\\
	&\stackrel{(a)}{=}
	n\log T_{\max}
	-
	\frac{n}{2}\log\epsilon_n
	-
	\frac{3n}{2}\log 2
	-
	\frac{n}{2}\log(\pi e)
	+
	o(n)
	\nonumber\\
	&=
	n\log T_{\max}
	-
	\frac{n}{2}\log\epsilon_n
	-
	\frac{n}{2}\log(8\pi e)
	+
	o(n),
	\label{Eq:Log_M_Volume_Bound}
\end{align}
where $(a)$ follows from $r_0=\sqrt{n\epsilon_n}$. For the scaling  $\epsilon_n=a\,n^{(b-1)/2},$ the preceding bound further reduces to
\begin{align}
	\log M_n
	\geq
	\frac{1-b}{4}\,n\log n
	+
	n\log
	\left(
	\frac{T_{\max}}{\sqrt{8\pi e a}}
	\right)
	+
	o(n),
	\label{Eq:Log_M_Compact}
\end{align}

Hence, the dominant growth of $\log M_n$ is of order $n\log n$. We therefore adopt the finite-blocklength normalized coding rate
\begin{equation}
	R_n
	\triangleq
	\frac{\log M_n}{n\log n}.
\end{equation}
Dividing \eqref{Eq:Log_M_Compact} by $n\log n$ gives
\begin{align}
	R_n
	&\geq
	\frac{1-b}{4}
	+
	\frac{1}{\log n}
	\log
	\left(
	\frac{T_{\max}}{\sqrt{8\pi e a}}
	\right)
	+
	o\!\left(\frac{1}{\log n}\right),
	\label{Ineq.R_n_LB}
\end{align}
and consequently
\begin{equation}
	\liminf_{n\to\infty} R_n
	\geq
	\frac{1-b}{4}.
	\label{Eq:R_n_Liminf}
\end{equation}

Hence, the desired lower bound on $M_n$ follows directly from the existence of a saturated packing. The proof does not
require any lower bound on the finite-domain packing density
$\Updelta_n(\mathscr{S})$, nor does it require transferring a
whole-space sphere-packing density result to the bounded hypercube.
The boundary issue is therefore avoided rather than estimated.

\textbf{\textcolor{mycolor12}{Encoding.}}
We assume that the encoding function is deterministic, i.e., each message $i \in [\![M_n]\!]$ is associated to a known codeword $\fc_i.$ Hence, $\forall i \in [\![M_n]\!],$ the transmitter sends $\fx = \fc_i.$

\textbf{\textcolor{mycolor12}{Decoding.}}
Let $e_1, e_2, \eta_0, \zeta_0, \zeta_1,b \hspace{-.6mm} > \hspace{-.6mm} 0$ and $a>0$ be arbitrarily small and fixed constants, respectively. Next, let:
\begin{itemize}[leftmargin=*]
	\item $Y_t(i) = c_{i,t} + Z_t,\, \forall t \in [\![n]\!]$ is channel output given $\fx=\fc_i.$
	\item $\fY(i)= (Y_t(i))_{t=1}^n$ is the output vector.\vspace{1mm}
	\item $\bar{\fZ} \triangleq \fZ-\mu\mathbf{1}=(\bar{Z}_t)_{t=1}^n$, where $\bar{Z}_t\triangleq Z_t-\mu$. 
	\item $\beta_1 \triangleq n^{-1} \big( \big\| \bar{\fZ} \big\|^2 \hspace{-.4mm} + \hspace{-.4mm} \big\| \fc_i - \fc_j \big\|^2 \big)$ and $\beta_2 \hspace{-.6mm} \triangleq \hspace{-.6mm} 2n^{-1} \hspace{-.4mm} \sum_{t=1}^{n} \big(c_{i,t} \hspace{-.3mm} - \hspace{-.3mm} c_{j,t} \big) \bar{Z}_t.$
	\item $\beta \triangleq \beta_1 + \beta_2 = n^{-1} \big\| \fY(i) - \fc_j - \mu \mathbf{1} \big\|^2$ and $\alpha \triangleq \mu^3 / \lambda.$
	\item $\E_0 \hspace{-.4mm} = \hspace{-.4mm} \{ \fZ \,:\, |\beta_2| > \delta_n \}.$
	\item $\E_1 \hspace{-.4mm} = \hspace{-.4mm} \big\{ \fZ \,:\, \beta_1 - \alpha \leq 2\delta_n \big\}.$
	\item $\E_2 \hspace{-.4mm} = \hspace{-.4mm} \big\{ \fZ \,:\, \beta - \alpha \leq \delta_n \big\}.$
	\item $T(\fy,\fc_j) \hspace{-.1mm} = \hspace{-.1mm} n^{-1} \| \fy - \fc_j - \mu \mathbf{1} \|^2 \hspace{-.1mm} - \hspace{-.1mm} \alpha$ is a \emph{decoding measure}.
	\item $\delta_n = (4a) / 3n^{(1 - b)/2}$ is \emph{decoding threshold}.
\end{itemize}
To identify if message $j \in [\![M_n]\!]$ was sent, the decoder checks whether $\mathbf{y}$ belongs to the decoding set:
\begin{align}
	\label{Eq.Dec_Meas0}
	\mathbbmss{D}_j = \Big\{ \fy \in \mathbb{R}^{n} \,:\; T(\fy,\fc_j) \leq \delta_n \Big\} .
\end{align}
\textbf{\textcolor{mycolor12}{Type I Error.}}
The type I errors occur when the transmitter sends $\fc_i,$ yet $\fY \notin \mathbbmss{D}_i.$ $\forall i \in [\![M_n]\!],$ the type I error probability reads
\begin{align}
	\label{Eq.TypeIError_Anal}
	P_{e,1}(i) = \Pr\big( \fY(i) \in \mathbbmss{D}_i^c \big) = \Pr\big( T(\fY(i),\fc_i) > \delta_n \big).
\end{align}
To bound $P_{e,1}(i),$ we perform Chebyshev's inequality, namely
\begin{align}
	\Pr\big( \big| T(\fY(i),\fc_i) - \mathbb{E} \big[ T(\fY(i),\fc_i) \big] \big| \hspace{-.5mm} > \hspace{-.5mm} \delta_n \big) \hspace{-.5mm} \leq \hspace{-.5mm} \frac{\text{Var} \big[ T(\fY(i),\fc_i) \big]}{\delta_n^2} .
\end{align}

We first calculate the expectation. 
\begin{align}
	\mathbb{E}\!\left[T(\mathbf Y(i),\mathbf c_i)\right]
	=
	\mathbb{E}\Big[
	\frac{1}{n}
	\sum_{t=1}^{n}(Z_t-\mu)^2
	-\alpha
	\Big]
	& =
	\frac{1}{n}
	\sum_{t=1}^{n}
	\mathbb{E}\big[\bar{Z}_t^2\big]
	-\alpha
	\nonumber\\&
	=
	\mathbb{E}\big[ \bar{Z}^2\big]
	-\alpha
	=
	\operatorname{Var}(Z)
	-\alpha
	=
	\frac{\mu^3}{\lambda}
	-\frac{\mu^3}{\lambda}
	=0,
\end{align}
where we use $\| \fY(i) - \fc_i \| = \| \fZ \|,$ linearity of the expectation, $Z_{t} \overset{\text{\tiny i.i.d.}}{\sim} \text{IG}(\mu,\lambda)$ with $\text{Var}[Z_{t}] = \mathbb{E}\,[\bar{Z}_t^2] = \mu^3 / \lambda.$ Next, we present a lemma which characterizes bound on the non-central moment of an IG RV.

\begin{customlemma}{2}[Fourth-Moment Growth Bound]
	\label{Lem.MGF}
	
	Let $Z\sim\mathrm{IG}(\mu,\lambda)$ be an IG distributed RV with mean $\mu$ and shape parameter $\lambda$. Then, the fourth centered moment of $Z$ satisfies
	\begin{equation}
		\mathbb{E}\big[\bar{Z}^4\big]
		\leq
		15\frac{\mu^6}{\lambda^2}
		\left(1+\frac{\mu}{\lambda}\right).
		\end{equation}
	\end{customlemma}
	\begin{proof}
		The proof is provided in Appendix \ref{App.MGF}.
	\end{proof}
	Now, we proceed to establish an upper bound on the variance of the decoding measure as follows
	\begin{align}
		\label{Ineq.Var_Decoding_Metric}
		\operatorname{Var}\!\left[T(\mathbf Y(i),\mathbf c_i)\right]
		=
		\operatorname{Var}\!\Big[
		\frac{1}{n}\sum_{t=1}^{n}\bar{Z}_t^2-\alpha
		\Big]
		& =
		\frac{1}{n^2}
		\sum_{t=1}^{n}
		\operatorname{Var}\!\Big[\bar{Z}_t^2\Big]
		\nonumber\\&
		\stackrel{(a)}{=}
		\frac{1}{n}
		\left[
		\mathbb{E}\big[\bar{Z}^4\big]
		-
		\mathbb{E}\big[\bar{Z}^2\big]^2
		\right]
		\nonumber\\&
		\stackrel{(b)}{=}
		\frac{1}{n}
		\left[
		\mathbb{E}\big[\bar{Z}^4\big]
		-
		(\mu^3/\lambda )^2
		\right]
		\nonumber\\&
		\stackrel{(c)}{\leq}
		\frac{1}{n}\mathbb{E}\big[\bar{Z}^4\big]
		\stackrel{(d)}{\leq}
		\frac{15\mu^6}{n\lambda^2}
		\left(1+\frac{\mu}{\lambda}\right) ,
	\end{align}
	where
	\begin{itemize}
		\item $(a)$ exploits $\operatorname{Var}(X)=\mathbb{E}[X^2]-\mathbb{E}[X]^2$ with $X=\bar{Z}^2$.
		
		\item $(b)$ uses $\mathbb{E}[\bar{Z}^2]=\operatorname{Var}(Z)=\frac{\mu^3}{\lambda}$ for $Z\sim\mathrm{IG}(\mu,\lambda)$.
		
		\item $(c)$ follows from $(\mu^3/\lambda )^2\geq0$.
		
		\item $(d)$ follows from Lemma \ref{Lem.MGF}, which gives
		\[
		\mathbb{E}[\bar{Z}^4]
		\leq
		15\frac{\mu^6}{\lambda^2}
		\left(1+\frac{\mu}{\lambda}\right).
		\]
	\end{itemize}
	
	Thereby,
	\begin{align}
		\label{Ineq.TypeI_Final}
		P_{e,1}(i)
		&\leq
		\Pr\left(
		\left|
		T(\fY(i),\fc_i)
		-\mathbb{E}\!\left[T(\fY(i),\fc_i)\right]
		\right|
		>\delta_n
		\right)
		\nonumber\\
		&\stackrel{(a)}{=}
		\Pr\left(
		\left|T(\fY(i),\fc_i)\right|
		>\delta_n
		\right)
		\nonumber\\
		&\stackrel{(b)}{\leq}
		\frac{
			\operatorname{Var}\!\left[T(\fY(i),\fc_i)\right]
		}{\delta_n^2}
		\nonumber\\
		&\stackrel{(c)}{\leq}
		\frac{
			15\mu^6
			\left(1+\frac{\mu}{\lambda}\right)
		}{
			n\lambda^2\delta_n^2
		}
		\nonumber\\
		&=
		\frac{
			135\mu^6
		}{
			16a^2\lambda^2
		}
		\left(1+\frac{\mu}{\lambda}\right)
		\frac{1}{n^b}
		\triangleq \eta_0,
	\end{align}
	where $(a)$ follows from $\mathbb{E}\!\left[T(\fY(i),\fc_i)\right]=0,$
	$(b)$ follows from Chebyshev's inequality, and $(c)$ follows from the variance bound
	\[
	\operatorname{Var}\!\left[T(\fY(i),\fc_i)\right]
	\leq
	\frac{15\mu^6}{n\lambda^2}
	\left(1+\frac{\mu}{\lambda}\right),
	\]
	together with
	\[
	\delta_n=\frac{4a}{3n^{(1-b)/2}}.
	\]

	Hence, $P_{e,1}(i) \leq \eta_0 \leq e_1$ holds for sufficiently large $n$ and arbitrarily small $e_1 > 0.$

	\textbf{\textcolor{mycolor12}{Type II Error.}}
	We examine type II errors, i.e., when $\fY \in \mathbbmss{D}_j$ while the transmitter sent $\fc_i$ with $i \neq j \,.$ Then, for every $i,j \in [\![M_n]\!],$ the type II error probability is given by
	\begin{align}
		P_{e,2}(i,j) = \Pr \big( T(\fY(i);\fc_j) \leq \delta_n \big) = \Pr(\E_2) .
		\label{Eq.Pe2G}
	\end{align}
	Next, to bound the event $\E_2$ we decompose the square norm given in the event $\E_2$ as follows
	\vspace{-.7mm}
	\begin{align}
		n^{-1} \big\| \fZ + \fc_i - \fc_j - \mu \mathbf{1} \big\|^2 = n^{-1} \big( \big\| \bar{\fZ} \big\|^2 + \big\| \fc_i - \fc_j \big\|^2 \big) + \beta_2
		\label{Eq.TypeII-3}
	\end{align}
	Next, to bound $\Pr\left(\E_0 \right)$ we apply Chebyshev's inequality
	\begin{align}
		\label{Ineq.Event_E0_1}
		\hspace{-2mm} \Pr(\E_0) \leq \frac{\sum_{t=1}^{n} \text{Var}\big[  \big( c_{i,t} - c_{j,t} \big) \bar{Z}_t \big]}{n^2(\delta_n / 2)^2} = \frac{4 \mu^3 \big\| \fc_i - \fc_j \big\|^2}{\lambda n^2\delta_n^2} , \hspace{-1mm}
	\end{align}
	where we use $\text{Var}[\bar{Z}_{t}] = \text{Var}[Z_{t}] = \mu^3 / \lambda.$ Now by the triangle inequality
	
	\begin{equation}
		\| \fc_i - \fc_j \|^2 \leq \big( \sqrt{n} \| \fc_i \|_{\infty} + \sqrt{n} \| \fc_j \|_{\infty} \big)^2 \leq 4nT_{\rm max}^2.
	\end{equation}
	Thereby,
	\begin{align}
		\label{Ineq.Event_E0_2}
		\Pr(\E_0) \leq \frac{4 \mu^3 \big\| \fc_i - \fc_j \big\|^2}{\lambda n^2\delta_n^2} \stackrel{(a)}{\leq} \frac{9 \mu^3 T_{\rm max}^2}{a^2\lambda n^{b}}
		\triangleq \zeta_0,
	\end{align}
	where we use $\delta_n = 4a / 3n^{(1- b)/2}.$ Now, given the complementary event $\E_0^c,$ we get $\beta_2 > - \delta_n.$ Next, applying the law of total probability to $\E_2$ over $\E_0$ and its complement $\E_0^c,$ we obtain $\Pr(\E_2) \leq \Pr\left( \E_0 \right) + \Pr\left( \E_1 \right),$ where we use $\E_2 \cap \E_0 \subset \E_0$ and $\Pr(\E_2 \cap \E_0^c) \leq \Pr (\E_1 )$ which is proved in the following:
	\begin{align}
		\Pr(\E_2 \cap \E_0^c) \leq \Pr \big( \big\{ \beta_1 - \alpha \leq 2 \delta_n  \big\}  \big) = \Pr\left(\E_1 \right),
	\end{align}
	where $(a)$ holds since $\beta = \beta_1 + \beta_2$ and $(b)$ uses $\delta_n-\beta_2\leq 2\delta_n$ conditioned on $| \beta_2 | \leq \delta_n.$ We now proceed with bounding $\Pr\left(\E_1 \right):$ The codebook construction gives $\| \fc_i - \fc_j \|^2 \geq 4 n\epsilon_n.$ Now, by the choice of $\epsilon_n$ and $\delta_n$ we have $- \big\| \fc_i - \fc_j \big\|^2 \le - 12an/3n^{( 1 - b) / 2} = -3n \delta_n.$ Therefore, we establish
	\begin{align}
		\label{Ineq.Event_E1}
		\Pr(\E_1)
		\leq
		\Pr\left(
		\frac{1}{n}\|\bar{\fZ}\|^2
		-\frac{\mu^3}{\lambda}
		\leq
		-\delta_n
		\right)
		& \stackrel{(a)}{\leq}
		\frac{
			\sum_{t=1}^n
			\operatorname{Var}[\bar Z_t^2]
		}{
			n^2\delta_n^2
		}
		\nonumber\\
		&\stackrel{(b)}{\leq}
		\frac{
			15\mu^6
			\left(1+\frac{\mu}{\lambda}\right)
		}{
			\lambda^2 n\delta_n^2
		}
		=
		\frac{
			135\mu^6
		}{
			16a^2\lambda^2
		}
		\left(1+\frac{\mu}{\lambda}\right)
		n^{-b}
		\triangleq
		\zeta_1,
	\end{align}
	where in $(a)$ we use Chebyshev's inequality and $(b)$ follows by similar arguments in \eqref{Ineq.Var_Decoding_Metric} and \eqref{Ineq.TypeI_Final}. Thus, \eqref{Ineq.Event_E0_2} and \eqref{Ineq.Event_E1}, yields
	\begin{align}
		P_{e,2}(i,j) \leq \Pr(\E_0) + \Pr(\E_1) \leq \zeta_0 + \zeta_1 \leq e_2,
	\end{align}
	hence, $P_{e,2}(i,j) \leq e_2$ holds for sufficiently large $n$ and arbitrarily small $e_2 > 0 .$ We have thus shown that $\forall e_1,e_2 > 0$ and sufficiently large $n,$ there is an $(n, M_n,\allowbreak e_1, e_2)$-code. This completes the achievability proof of Proposition~\ref{Propos.Achievability}.
	
	\subsection{Upper Bound (Converse Proof)}
	\label{Subsec.Converse}
	
	The converse proof consists of two steps. First, we establish a minimum-separation property for reliable identification codebooks given in Lemma \ref{Lem.Converse}. The key idea is to show that if two distinct codewords are sufficiently close in every coordinate, then their induced output distributions become asymptotically indistinguishable. Consequently, the sum of the corresponding type-I and type-II error probabilities approaches one, contradicting the vanishing-error criterion.
	
	To formalize this argument, we exploit the near-zero tail behavior of the IG noise distribution in Lemma \ref{Lem.AS_LB}. Specifically, for sufficiently small \(\varepsilon>0\), there exists a constant \(C>0\) such that
	\[
	\Pr(Z_t\le\varepsilon)
	\le C\sqrt{\varepsilon}
	\exp\!\left(-\frac{\lambda}{2\varepsilon}\right).
	\]
	Combining this estimate with the union bound and the first Borel--Cantelli lemma yields an almost-sure lower bound on the minimum noise sample. We then use this bound to control the output-density ratio for sufficiently close codewords in Lemma \ref{Lem.Log_Likelihood_Ratio}, thereby establishing the separation property by contradiction.
	
	Second, we use the resulting separation property to derive a packing bound on the codebook size \(M_n\). After normalization by \(n\log n\), this bound yields the desired upper bound on the identification rate.

	In the following, to enable concise derivations and to support the subsequent analytical development, we introduce the following notational conventions and assumptions:
	\begin{itemize}[leftmargin=*]
		\item Let $(\Omega,\mathcal{F},\Pr)$ be a probability space supporting an i.i.d.\ sequence of IG RVs $\{Z_t\}_{t\ge1}$, each with distribution $\mathrm{IG}(\mu,\lambda)$, where $\mu,\lambda>0$ are fixed.
		\item $Y_t(i) = c_{i,t} + Z_t,\, \forall t \in [\![n]\!]$ denote the channel output at time $t$ \emph{conditioned} $\fx=\fc_i$ was sent.
		\item Let $\mathbbmss{C}_{\text{\tiny conv}}$ denote the admissible codebook considered in the converse proof, defined as $$\mathbbmss{C}_{\text{\tiny conv}} \triangleq \big\{ \fc_i \in \mathbb{R}^n:\; c_{i,t} \in [0 ,T_{\rm max}] , \forall \, i \in [\![M_n]\!],\, \forall t \in [\![n]\!] \big\}.$$
		\item The separation threshold is defined as $\alpha_n \triangleq a^2/n^{1+2b}$ with $a,b>0$ being an arbitrarily small constants, respectively.
		\item Let $\boldsymbol{\alpha}_n \triangleq (\alpha_1,\ldots,\alpha_n)$ denote the vector of the first $n$ terms of the sequence
		$\{\alpha_t\}_{t\ge1}$.
		\item Define the minimum noise sample among the first $n$ channel uses as $Z_{\min}(n)\triangleq \min_{1\le t\le n} Z_t.$ For every $\eta>1$, define $c_\eta\triangleq \lambda/(2(1+\eta)).$
	\end{itemize}

	We first formalize the almost-sure lower bound on the minimum noise sample required for the density-ratio analysis.

		\begin{customlemma}{3}[Almost-Sure Minimum Noise Bound]
		\label{Lem.AS_LB}
		Almost surely, for all sufficiently large $n$,
		\begin{equation}
			Z_{\min}(n)\ge \frac{c_\eta}{\ln n},
		\end{equation}
		or equivalently, $Z_{\min}(n)=\Omega(1/\log n).$ Consequently, almost surely,
		\begin{equation}
		\frac{1}{Z_{\min}(n)}=O(\ln n),
		\qquad
		\frac{1}{Z_{\min}^2(n)}=O(\ln^2 n).
		\end{equation}
	\end{customlemma}
	
	\begin{proof}
		The proof is provided in Appendix \ref{App.AS_LB}.
	\end{proof}
	
	The preceding lemma bounds how closely noise samples can approach zero. We use this property to establish asymptotic indistinguishability for codewords whose coordinates differ by less than the prescribed threshold. Before stating the lemma, we define the set of admissible distinct codeword pairs $\mathcal P_n
	\triangleq
	\left\{
	(\mathbf c_{i_1},\mathbf c_{i_2})
	\in \mathbbmss{C}_{\text{\tiny conv}}^2
	:\;
	i_1,i_2\in[\![M_n]\!],\;
	i_1\neq i_2,\;
	|c_{i_1,t}-c_{i_2,t}|\le 2\alpha_n,\;
	\forall t\in[\![n]\!]
	\right\}.$

	\begin{customlemma}{4}[Asymptotic Output-Density Ratio Bound]
		\label{Lem.Log_Likelihood_Ratio}
		Suppose that $(\mathbf c_{i_1},\mathbf c_{i_2})\in\mathcal P_n$ and that the coordinatewise codeword difference satisfies $|c_{i_1,t}-c_{i_2,t}|\le 2\alpha_n, t\in[\![n]\!],$ where $\alpha_n=a^2/n^{1+2b}$, with constants $a,b>0$. Then, the output-density ratio converges to one, i.e.,
		\begin{equation}
			\lim_{n\to\infty}
			\left|
			\frac{f_{\mathbf Z}(\mathbf y-\mathbf c_{i_2})}
			{f_{\mathbf Z}(\mathbf y-\mathbf c_{i_1})}
			-1
			\right|=0.
		\end{equation}
		Moreover, the convergence is uniform over all admissible codeword pairs on the good set.
	\end{customlemma}

	\begin{proof}
		The proof is provided in Appendix~\ref{App.Log_likelihood_Ratio}.
	\end{proof}

	We establish the separation property by contradiction: if two distinct codewords were within \(2\alpha_n\) in every coordinate, Lemma~\ref{Lem.Log_Likelihood_Ratio} would imply asymptotically indistinguishable output distributions, contradicting vanishing identification error.

	Assume that $R$ is an achievable identification-rate for $\bIG$ under the vanishing-error criterion. Then, there exists a sequence of $(n,M_n,e_1,e_2)$-identification codes, denoted by $(\mathbbmss{C}_{\text{\tiny conv}}^{(n)},\mathbbmss{D}^{(n)})_{n\in\mathbb{N}}$, whose maximum type I and type II error probabilities satisfy $\bar{P}_{e,1}^{(n)} \to 0$ and $\bar{P}_{e,2}^{(n)} \to 0$ as $n \to \infty.$ 
	
	\begin{customlemma}{5}[Separation Property of Identification Codebooks]
		\label{Lem.Converse}
		For all sufficiently large codeword length $n$, the associated codebook $\mathbbmss{C}_{\text{\tiny conv}}^{(n)}$ satisfies the following property: for any two distinct codewords $\mathbf{c}_{i_1}$ and $\mathbf{c}_{i_2}$ with indices $i_1, i_2 \in [\![M_n]\!]$ and $i_1 \neq i_2$, there exists at least one coordinate $t' \in [\![n]\!]$ such that the magnitude of the difference between the corresponding symbols obeys
		\begin{align}
			\label{Ineq.Conv_Distance}
			| c_{i_1,t'} - c_{i_2,t'} | \geq 2\alpha_n.
		\end{align}
	\end{customlemma}
	\begin{proof}
		The proof is provided in Appendix~\ref{App.Converse_Proof}.
	\end{proof}
	
	With the separation property established, we apply Lemma~\ref{Lem.Converse} to derive a packing bound on the codebook size and hence an upper bound on the identification capacity. Since the minimum distance between distinct codewords in \(\mathbbmss{C}_{\text{\tiny conv}}\) is at least \(2\alpha_n\), we can arrange non-overlapping spheres $\S_{\fc_i}(n,\alpha_n)$ whose centers belong to the codebook $\mathbbmss{C}_{\text{\tiny conv}}.$ Therefore, it follows that the number of codewords $M,$ is bounded by
	\begin{align}
		\label{Ineq.Codebook_Size_UB}
		M_n = \frac{\text{Vol}\left(\bigcup_{i=1}^{M_n} \S_{\fc_i}(n,\alpha_n) \right)}{\text{Vol}(\S_{\fc_1}(n,\alpha_n))}
		\stackrel{(a)}{\leq}
		\frac{\text{Vol}\big( \mathbbmss{Q}_{\f0}(n,T_{\rm max} + 2\alpha_n) \big)}{\text{Vol}(\S_{\fc_1}(n,\alpha_n))}
		\stackrel{(b)}{\leq} \frac{(T_{\rm max}+2\alpha_n)^{n}}{\text{Vol}(\S_{\fc_1}(n,\alpha_n))},
	\end{align}
	where $(a)$ holds since
	\begin{align}
		\mathbbmss{C}_{\text{\tiny conv}} \subseteq \mathbbmss{Q}_{-(\boldsymbol{\alpha}_n,\boldsymbol{\alpha}_n)}(n,T_{\rm max} + 2\alpha_n)
		= \big\{ \fc_i \in \mathbb{R}^{n} \hspace{-.4mm}: \hspace{-.2mm} - \alpha_n \leq c_{i,t} \leq T_{\rm max} + \alpha_n, \, \forall \, i , \forall \, t \big\}
	\end{align}
	which implies $$\text{Vol}( \mathbbmss{C}_{\text{\tiny conv}} ) \hspace{-.4mm} \leq \hspace{-.4mm} \text{Vol}( \mathbbmss{Q}_{-(\boldsymbol{\alpha}_n,\boldsymbol{\alpha}_n)}(n,T_{\rm max} + 2\alpha_n) )$$ and $(b)$ uses $\text{Vol}( \mathbbmss{Q}_{-(\boldsymbol{\alpha}_n,\boldsymbol{\alpha}_n)}(n,T_{\rm max} + 2\alpha_n) ) \hspace{-.4mm} = \hspace{-.4mm} (T_{\rm max} + 2\alpha_n)^{n}.$ To simplify \eqref{Ineq.Codebook_Size_UB}, we employ the upper bound on negative log-volume of a hypersphere given in Lemma \ref{Lem.Neg_Log_Vol} and obtain the following upper bound on the logarithm of the achievable codebook size,

	\begin{align}
		\label{Ineq.Log_M_UB}
		\log M_n
		& \leq
		n\log(T_{\rm max}+2\alpha_n)
		-\log\big(\operatorname{Vol}(\mathcal S_{\fc_1}(n,\alpha_n))\big)
		\nonumber\\&
		\leq
		n\log(T_{\rm max}+2\alpha_n)
		+\frac n2\log\frac n2
		-n\log\alpha_n
		-\frac n2\log(\pi e)
		+o(n).
	\end{align}
	
	Now, let
	\[
	\alpha_n=\frac{a^2}{n^{1+2b}},
	\qquad b>0.
	\]
	Since $b>0$, we have $\alpha_n\to0$. Assuming $T_{\rm max}>0$ is fixed,
	\begin{align}
		\log(T_{\rm max}+2\alpha_n)
		=
		\log T_{\rm max}
		+
		\log\left(1+\frac{2\alpha_n}{T_{\rm max}}\right)
		=
		\log T_{\rm max}+O(\alpha_n)
		=
		\log T_{\rm max}+o(1).
	\end{align}
	Therefore, $n\log(T_{\rm max}+2\alpha_n)
	=
	n\log T_{\rm max}+o(n).$
	Furthermore,
	\begin{align}
		-n\log\alpha_n
		=
		-n\log\left(a^2n^{-(1+2b)}\right)
		=
		-n\log(a^2)
		+(1+2b)n\log n.
	\end{align}
	
	Substituting these relations into \eqref{Ineq.Log_M_UB} gives
	\begin{align}
		\label{Ineq.Log_M_UB2}
		\log M_n 
		& \leq 
		n\log T_{\rm max} -n\log a^2
		+(1+2b)n\log n
		+\frac n2\log\frac n2
		-\frac n2\log(\pi e)
		+o(n)
		\nonumber\\
		&=
		\left(\frac32+2b\right)n\log n
		+
		n\log\left(
		\frac{T_{\rm max}}{a^2\sqrt{2\pi e}}
		\right)
		+o(n).
	\end{align}
	
	The dominant term in \eqref{Ineq.Log_M_UB2} is of order $n\log n$. Consequently, we adopt the normalized coding rate defined in \eqref{Eq.Normalized_Coding_Rate},
	\[
	R_n
	\triangleq
	\frac{\log M_n}{n\log n},
	\]
	which captures the dominant growth of the codebook size and yields a finite asymptotic coding rate. Dividing \eqref{Ineq.Log_M_UB2} by $n\log n$, we obtain
	\begin{equation}
		\label{Ineq.R_n_UB}
		R_n
		\leq
		\frac32+2b
		+
		\frac{1}{\log n}
		\log\left(
		\frac{T_{\rm max}}
		{a^2\sqrt{2\pi e}}
		\right)
		+
		\frac{o(n)}{n\log n}.
	\end{equation}
	
	Since
	\[
	\frac{1}{\log n}
	\log\left(
	\frac{T_{\rm max}}
	{a^2\sqrt{2\pi e}}
	\right)
	\longrightarrow 0
	\qquad\text{and}\qquad
	\frac{o(n)}{n\log n}\longrightarrow0,
	\]
	it follows that
	\begin{equation}
	R_n
	\leq
	\frac32+2b+o(1).
	\end{equation}
	Thus,
	\begin{equation}
	\limsup_{n\to\infty}R_n
	\leq
	\frac32+2b.
	\end{equation}
	Since $b>0$ is arbitrary, letting $b\downarrow0$ yields
	\begin{equation}
	\limsup_{n\to\infty}R_n
	\leq
	\frac32.
	\end{equation}
	This completes the proof of Proposition~\ref{Propos.Converse}.

	\section{Conclusions, Discussions and Research Outlook}
	\label{Sec.Conclusion}
	
	This work provides a rigorous investigation of the identification problem for the IG channel, a canonical model that fits the diffusion-based MC channels with drift and is motivated by real physical and biological transport processes where particles move randomly and are carried by a flow. We show that reliable identification can be achieved with a super-exponentially large codebook, namely, $M = 2^{\Theta(n \log n)},$ and derive both lower and upper bounds on the achievable rate $R.$ The theoretical framework developed here suggests two promising directions for future research, including:
	\begin{itemize}[leftmargin=*]
		\item \textbf{\textcolor{mycolor12}{Multiple Channel Uses.}} In this paper, we consider an MC system in which only a single molecular type is employed during each channel use. Specifically, each channel use corresponds to the release of one individual molecule, with the release time representing the transmitted input symbol and the corresponding molecule arrival time at the receiver constituting the channel output. Future work should make the model more realistic by allowing multiple molecule transmissions in the same environment instead of assuming clean, isolated channel uses. A key question is how often the channel can be reused. If molecules are indistinguishable, the transmitter must wait until all molecules from one message arrive before sending the next one leading to slow communication and a delayed decoding procedure. On the other hand, if molecules are distinguishable, different transmissions can overlap in time because the receiver can tell which molecule belongs to which message. This allows faster, even simultaneous communication. This phenomenon resembles the spatial orthogonality (spatial channel use) used for identification in MC systems over the Poisson channel \cite{Salariseddigh-TMBMC} the receiver employs distinct receptors, each specifically designed to bind a single type of molecule, with no affinity for others or the affine Poisson channel \cite{Salariseddigh25_ITW} where receptors are not perfectly selective, but they operate in an affinity-based regime, namely, each receptor is tuned for a target molecule but may also bind other molecules with lower probability. Analyzing such affinity-based receptors, together with determining the eligible classes of the cross-reactivity affinity matrices is a potential direction for future work. Future research should study how to design systems that handle overlapping transmissions and how much performance improves when molecules can be distinguished.

		\item \textbf{\textcolor{mycolor12}{Inter-Symbol Interference.}}
		In the present model, each molecule carries a detectable identifier,
		and propagation delays are independent. Consequently, although
		molecules released during different channel uses may overlap in
		transit and arrive out of chronological order, the receiver can
		associate each arrival with its corresponding channel use. Thus,
		arrival reordering does not, by itself, introduce ISI, and the identity-indexed channel remains
		memoryless under the assumed propagation model. An important direction for future work is to extend the framework
		to settings in which channel memory arises explicitly, for example,
		when molecules are indistinguishable and arrivals are observed
		without individual labels, or when interactions between molecules
		or with the propagation environment affect subsequent observations.
		The unbounded support of the IG propagation delay, which implies
		$\Pr(Z>T_s)>0$ for every finite $T_s$, motivates investigating
		such scenarios, although unbounded delays alone do not induce
		memory in the present model. Addressing these settings will require
		analytical and identification frameworks that go beyond the
		memoryless setting and may draw on established models for channels
		with state and memory, such as the Gilbert--Elliott model.

		\item \textbf{\textcolor{mycolor12}{Noise Regularity Condition.}} In the converse proof, the continuity of the IG PDF is established under a standard regularity assumption on the noise process. We emphasize that a deterministic or uniform lower bound on the noise sequence is neither required nor can be imposed, since the IG noise has support on \(\mathbb{R}_{+}\) and therefore admits arbitrarily small realizations with positive probability. Instead, we exploit an intrinsic property of the IG distribution, namely its super-exponential decay near the origin. This tail behavior ensures that extremely small noise realizations occur with sufficiently rapidly vanishing probability. In particular, we employ the Borel–Cantelli lemma to translate this rapid decay into an almost sure statement, thereby guaranteeing that events corresponding to excessively small noise realizations occur only finitely often. This yields an almost sure lower bound on the minimum noise realization of the form
		\begin{equation}
			\Pr\!\left(
		Z_{\min}(n)
		\ge
		\frac{c_\eta}{\ln n}
		\ \text{for all sufficiently large } n
		\right)
		=1,
		\end{equation}
		for every \(\eta>1\), where \(c_\eta = \lambda / (2(1+\eta)) > 0\), or equivalently,
		\[
		Z_{\min}(n)
		=
		\Omega\!\left(\frac{1}{\ln n}\right)
		\qquad \text{almost sure.}
		\]
		
		Under this almost sure control of the minimum noise, the continuity of the IG PDF required in the converse analysis holds on an event of probability one, and the converse argument proceeds unchanged.

		\item \textbf{\textcolor{mycolor12}{Decoder.}} In an information-theoretic setting, the IG channel differs fundamentally from the Gaussian channel in both noise structure and the resulting decoding rule. The Gaussian model assumes additive, amplitude-domain noise that is symmetric and signal-independent with constant variance, yielding a log-likelihood that depends on a simple quadratic (Euclidean) distance and leads to standard ML decoding. Specifically, observe that the log-likelihood for the Gaussian channel reads
		\begin{equation}
			\log f_{Y \mid X}(y \mid x)
			=
			- \frac{(y - x)^2}{2\sigma^2}
			- \frac{1}{2} \log(2\pi\sigma^2)
		\end{equation}
		which features a clean quadratic Euclidean distance. 
		
		By contrast, the IG channel models noise as a strictly positive, additive in the time domain and skewed random propagation delay. Under the adopted additive model, $Y=X+Z$, where $Z$ follows an IG distribution with fixed mean $\mu$ and shape parameter $\lambda$. Accordingly, the noise statistics are independent of the input, with $\mathbb{E}[Z]=\mu$ and $\operatorname{Var}(Z)=\mu^3/\lambda$. Despite this additive structure, the asymmetry and one-sided support of the IG noise distribution lead to a fundamentally different likelihood function. This distinction has direct implications for decoding. Specifically, defining $z\triangleq y-x$, the IG log-likelihood for $z>0$ is
		\begin{equation}
			\log f_{Y|X}(y|x)
			=
			-\frac{3}{2}\log z
			-\frac{\lambda}{2}
			\left(
			\frac{z}{\mu^2}
			+\frac{1}{z}
			\right)
			+ C,
		\end{equation}
		where
		\begin{equation}
			C \triangleq \frac{\lambda}{\mu} + \frac{1}{2}\log\left(\frac{\lambda}{2\pi}\right)
		\end{equation}
		is independent of $x$ and $y$. The conditional density is zero for $y\le x$, reflecting the strictly positive support of the propagation delay. Unlike the Gaussian log-likelihood, the IG log-likelihood induces a nonlinear, asymmetric, non-Euclidean decoding metric that is nonquadratic in the residual $y-x,$ owing to its logarithmic and reciprocal terms. Consequently, ML decoding for the IG channel requires maximizing the IG conditional likelihood over admissible input candidates, rather than simply minimizing the squared Euclidean distance. Both optimal and mismatched decoders must account for the nonquadratic likelihood structure, asymmetric noise distribution, and one-sided support of the IG channel, rather than relying solely on conventional Euclidean distance-based rules.

		In the Shannon channel coding theorem setting, optimal decoding is explicitly tied to the channel law via the likelihood function. While Gaussian channels yield a Euclidean distance metric, the IG channel induces a nonlinear, non-Euclidean, and asymmetric metric. Therefore, applying Euclidean (nearest-neighbor) decoding to the IG channel constitutes mismatched decoding, and is generally suboptimal because it does not reflect the intrinsic statistics of the channel. We note that for the identification problem, the optimal decision rule may depend on likelihood ratios or other criteria (see Neyman-Pearson lemma \cite{Lehmann22}), and Euclidean structure may still emerge as effective or even optimal under certain embeddings or approximations. Moreover, while the non-Euclidean nature of the IG likelihood implies a mismatched Euclidean decoding in the Shannon sense, this does not automatically extend to identification problems, where optimality criteria and induced metrics can differ.

		\item \textbf{\textcolor{mycolor12}{Converse Techniques.}} In our converse proof, we follow the standard approach used for the Gaussian channels \cite{Salariseddigh_IT,Salariseddigh_22_ITW} and Poisson channels \cite{Salariseddigh-TMBMC,Salariseddigh_OJCOMS_23}, which relies on the continuity of the channel law to establish a separation between distinct codewords. Within this framework, satisfying the continuity requirement typically forces the minimum distance between distinct codewords (or symbols) to vanish as the codeword length $n$ grows, which in turn leads to a looser upper bound. A more refined approach, originating from quantum information theory is based on the classical \emph{Fuchs–van de Graaf} inequalities \cite{Wilde13}, i.e.,
		\begin{align}
			1 - F(f_{\fZ}(.|\fc_{i_1}),f_{\fZ}(.|\fc_{i_2})) & \le d_{\rm TV}(f_{\fZ}(.|\fc_{i_1}) - f_{\fZ}(.|\fc_{i_2}) ) 
			\nonumber\\&
			\leq \sqrt{1 - F^2(f_{\fZ}(.|\fc_{i_1}),f_{\fZ}(.|\fc_{i_2}))},
		\end{align}
		where $d_{\rm TV}(.,.)$ denotes the total variation distance of output distribution for distinct codewords and $F(.,.)$ is the \emph{fidelity} or \emph{Bhattacharyya coefficient}. We note that the total variation distance associated with any sequence of \emph{reliable} identification codes remains uniformly bounded away from zero. This property accompanying with leveraging the above relation may enable us the construction of a sphere-packing argument whose spheres radii remains constant and does not decay with $n.$ Consequently, this approach may yield a tighter upper bound. However, we note that unlike the Gaussian and Poisson channels where their log-fidelity feature Euclidean \cite{Colomer26} or root square Euclidean distance \cite{Colomer_2025}, IG model suffers from a non-linear exponent which prevents straightforward linearization which is required for establishing mutual distance of codewords.
	\end{itemize}

	\subsection{Gap Between Achievability and Converse Bounds}
	
	The gap between the achievability and converse bounds, $$\frac{1}{4}\leq \mathbb{C}_{\rm I}(\bIG)\leq \frac{3}{2},$$ reflects limitations in both the current achievability and converse techniques. The achievability bound relies on a specific sphere-packing construction and a relatively simple distance-based decoder, which may restrict performance compared with more general code constructions and more sophisticated decoding strategies. For the converse, following standard approaches for Gaussian and Poisson channels, we rely on the continuity of the channel law to establish codeword separation. However, this approach typically forces the minimum distance to vanish with the blocklength, yielding a loose upper bound. A potentially sharper approach based on the Fuchs--van de Graaf inequalities could exploit the fact that the total variation distance of reliable identification codes remains bounded away from zero, potentially enabling a sphere-packing argument with non-vanishing radii. Unlike Gaussian and Poisson channels, however, the IG channel has a non-Euclidean fidelity structure whose nonlinear exponent prevents straightforward linearization into a tractable codeword-distance measure. Moreover, controlling the minimum noise almost surely via the Borel--Cantelli lemma addresses the behavior of the noise but does not resolve this geometric obstacle. Consequently, the current converse cannot yet leverage the fidelity-based approach to obtain a tighter upper bound, while the restricted achievability construction may also leave performance gains unexplored. Bridging the gap will therefore likely require both stronger converse techniques that account for the non-Euclidean structure of the IG channel and more general achievability constructions with improved decoding strategies.

	\subsection{Comparative Analysis of Channel Models}
	
	The identification problem studied here shares similarities with the discrete-time Poisson molecular
	channel considered in \cite{Salariseddigh-TMBMC}. Under their
	respective channel assumptions and input constraints, both models
	exhibit super-exponential identification codebook growth of the form $M_n = 2^{(n\log n)R},$
	with the reported bounds $1/4 \leq \mathbb{C}_{\mathrm{DI}} \leq 3/2.$ This correspondence suggests that super-exponential identification
	is not restricted to molecule-counting channels. However, the
	achievable scaling depends on the admissible input space, its
	effective resolution, and the channel distinguishability conditions. Both the Poisson and IG channel models employ packing-based constructions to establish identification capacity bounds. While the Poisson model combines packing with concentration or moment bounds, the IG model follows the same general approach but incorporates propagation-noise bounds and simpler distance-based decoding. Under suitable assumptions, both constructions can achieve $M_n = 2^{\Theta(n\log n)}$.
	
	\subsubsection{Channel-Specific Likelihood Geometry}
	
	The Poisson and IG converse arguments differ in their likelihood structures. For the Poisson channel, $Y\sim\operatorname{Pois}(\rho X+\lambda),$ for $X\geq0, \lambda>0,$ the conditional log-likelihood is
	\begin{equation}
		\log W(y|x)
		=-(\rho x+\lambda)
		+y\log(\rho x+\lambda)-\log(y!).
	\end{equation}
	Since $\rho x+\lambda\geq\lambda>0,$ the input-dependent parameter is uniformly bounded away from
	zero. The Poisson converse can therefore establish codeword
	separation through the input-dependent likelihood and a subsequent
	geometric packing argument, without requiring a sample-path lower
	bound on the minimum noise realization. In contrast, the IG channel satisfies $Y=X+Z,$ with $Z\sim\operatorname{IG}(\mu,\lambda),$ with conditional density, for $y>x$,
	\begin{equation}
		f_{Y|X}(y|x)
		=
		\sqrt{\frac{\lambda}{2\pi(y-x)^3}}
		\exp\left[
		-\frac{\lambda((y-x)-\mu)^2}
		{2\mu^2(y-x)}
		\right].
	\end{equation}
	Its log-likelihood contains terms proportional to
	$\log(y-x)$ and $(y-x)^{-1}$. Since $Y-X=Z$ can be arbitrarily
	close to zero, these terms cannot be controlled by the input
	constraint alone. The IG converse therefore requires control of
	the minimum noise realization
	\begin{equation}
		Z_{\min}(n)=\min_{1\leq t\leq n}Z_t.
	\end{equation}

	\subsubsection{Near-Zero Tail and the Borel--Cantelli Argument}
	
	The IG density decays exponentially near zero, with dominant
	behavior $\exp(-\lambda/(2z))$. Hence, for sufficiently small
	$\varepsilon>0$,
	\[
	\Pr(Z_t\leq\varepsilon)
	\leq C\exp\left(-\frac{\lambda}{2\varepsilon}\right).
	\]
	By the union bound,
	\[
	\Pr(Z_{\min}(n)\leq\varepsilon)
	\leq Cn\exp\left(-\frac{\lambda}{2\varepsilon}\right).
	\]
	Choosing $\varepsilon_n=\lambda/[2(1+\eta)\ln n]$ for
	any $\eta>1$ yields $$\Pr(Z_{\min}(n)\leq\varepsilon_n)\leq Cn^{-\eta}.$$
	Since these probabilities are summable, the first
	Borel--Cantelli lemma implies that, almost surely, $Z_{\min}(n)=\Omega\left(1 / \log n \right).$ This almost-sure lower bound controls inverse powers of the
	minimum noise realization, preventing the singular terms in
	the likelihood ratio from growing too rapidly. It thus supports
	the likelihood comparison and codeword-separation arguments
	in the IG converse without requiring a deterministic lower
	bound on every noise sample.

	\subsubsection{Common Geometric Endpoint of the Converse}
	
	Despite their different likelihood structures, both converse
	arguments ultimately establish codeword separation and apply
	geometric packing. For the Poisson channel, a relative
	symbol-wise separation condition takes the form
	\begin{equation}
		\left|
		1-\frac{\rho u_{i_2,t}+\lambda}
		{\rho u_{i_1,t}+\lambda}
		\right|>\epsilon_n'
	\end{equation}
	for at least one coordinate $t$ for every pair of distinct
	codewords. For the IG channel, the corresponding separation
	argument additionally requires the near-zero-tail control
	established above. Thus, both converses follow the general
	structure

	\begin{equation}
		\begin{array}{l}
			\text{likelihood analysis}\\[3pt]
			\qquad\qquad \searrow\\[3pt]
			\qquad\qquad\qquad\qquad
			\text{codeword separation}\\[3pt]
			\qquad\qquad\qquad\qquad \searrow\\[3pt]
			\qquad\qquad\qquad\qquad\qquad\qquad
			\text{geometric packing}\\[3pt]
			\qquad\qquad\qquad\qquad\qquad\qquad \searrow\\[3pt]
			\qquad\qquad\qquad\qquad\qquad\qquad\qquad\qquad
			\text{codebook-size bound}.
		\end{array}
	\end{equation}

	\subsubsection{Counting versus Timing Molecular Communication}
	
	The two models also represent different communication mechanisms.
	The Poisson channel describes molecule-counting communication,
	whereas the IG channel models molecular timing through a
	first-passage-time distribution. Their common super-exponential
	identification scaling suggests that this phenomenon is not
	exclusive to molecule counting. Nevertheless, the IG converse
	requires additional probabilistic control because its likelihood
	contains singular terms involving the random residual $Z=Y-X$.
	The principal distinction is therefore that the Poisson likelihood
	has an input-controlled parameter bounded away from zero, whereas
	the IG likelihood requires an almost-sure lower bound on the
	minimum propagation delay before geometric packing can be applied.

	\subsection{Open Problems}

	The results developed in the preceding sections lead naturally to several questions that remain unresolved by the present analysis. In this section, we focus on one direction that appears particularly promising for further investigation: the affine generalization of the IG channel.

	\subsubsection{Affine IG Channel}
	
	This paper considers a single molecular type per channel use.
	A natural extension is to employ multiple distinguishable molecule
	types to transmit parallel information streams through
	molecular type-division multiplexing, analogous to spatial
	multiplexing in MC systems with selective
	receptors \cite{Salariseddigh-TMBMC}. More generally, affinity-based
	receptors \cite{Salariseddigh25_ITW} may exhibit cross-reactivity:
	each receptor preferentially responds to certain molecule types
	but may also respond to others with lower affinity. Characterizing
	the identification capacity of such multi-type molecular channels
	and understanding how the affinity structure affects decoding
	performance remain open problems.
	
	As a preliminary model, motivated by biological olfaction, we
	consider $N$ molecule types and $K$ receptor types, with
	$K=o(N)$. The resulting affine IG model is
	\begin{align}
		Y_{k,n}=a_{k,n}x_n+Z_{k,n},
	\end{align}
	for $k\in[\![K]\!]$ and $n\in[\![N]\!]$, where $x_n$ denotes the
	release time of molecule type $n$, $a_{k,n}$ is the affinity
	coefficient between receptor $k$ and molecule type $n$, and
	$Z_{k,n}\sim\operatorname{IG}(\mu_{k,n},\lambda_{k,n})$ models
	the propagation delay. Here, $Y_{k,n}$ denotes the corresponding
	arrival-time observation at receptor $k$. Defining
	$\mathbf{x}=[x_1,\ldots,x_N]^T$,
	$\mathbf{A}=[a_{k,n}]$,
	$\mathbf{X}=\mathbf{1}_K\mathbf{x}^T$,
	$\mathbf{Y}=[Y_{k,n}]$, and $\mathbf{Z}=[Z_{k,n}]$, the model
	can be written compactly as
	\begin{align}
		\mathbf{Y}=\mathbf{A}\odot\mathbf{X}+\mathbf{Z},
	\end{align}
	where $\odot$ denotes the Hadamard product and $\mathbf{A}$ is
	the affinity matrix.
	
	An important open question is to characterize the identification
	capacity of this model and determine how it depends on the number
	of molecule types, the number of receptors, the affinity matrix,
	and the propagation-noise parameters. In particular, whether
	structural properties of $\mathbf{A}$, such as its rank, provide
	useful bounds on the achievable identification codebook size
	remains to be established. Resolving this question requires a
	capacity analysis that accounts for the effective distinguishability
	of the admissible inputs under the affine channel model.

	\section{Acknowledgments}
	The authors acknowledge Professor Dr. -Ing. Dr. rer. nat. Holger Boche for helpful discussions. Dr. -Math. Christian Deppe was supported by the Federal Ministry of Research, Technology, and Space (BMFTR) under the research program Communication Systems – “Souverän. Digital. Vernetzt.” through the joint project 6G-life (Project ID: 16KIS2415) and the project QSTARS (Project ID: 16KIS2196).

	\appendices

	\section{Density of Brownian Motion with Drift}
	\label{App.PDF_Pos}
	
	We consider the process $X_t = vt + \sigma B_t,$ where \(B_t\) is the standard Brownian motion \cite{Karatzas14}. A fundamental property of Brownian motion is $B_t \sim \mathcal{N}(0,t).$ Next, invoking the scaling property gives $\sigma B_t \sim \mathcal{N}(0,\sigma^2 t).$ Now, observe that if \(Z \sim \mathcal{N}(0,\sigma^2 t)\), then adding a constant drift impacts the expectation by $vt$ factor, i.e., $Z + vt \sim \mathcal{N}(vt,\sigma^2 t).$ Hence, $X_t \sim \mathcal{N}(vt,\sigma^2 t).$ Next, observe that if \(X \sim \mathcal{N}(\mu,\tau^2)\), then its density is given by
	\begin{align}
		f(x) = \frac{1}{\sqrt{2\pi \tau^2}} \exp\!\left(-\frac{(x-\mu)^2}{2\tau^2}\right).
	\end{align}
	
	Now, substituting \(\mu = vt\) and \(\tau^2 = \sigma^2 t\), we obtain
	\begin{align}
		f_X(x,t) = \frac{1}{\sqrt{2\pi\sigma^2 t}} \exp\!\left(-\frac{(x-vt)^2}{2\sigma^2 t}\right).
	\end{align}

	\section{Equivalence of Inverse Gaussian and First-Passage Time PDF}
	\label{App.Eq_IG_FPT}
	
	In the following, we show that the first-passage time PDF can be derived from the IG PDF by adopting appropriate mean and shape parameters. We start with the IG density:
	\begin{align}
		\label{Eq.Inv_Gau_PDF_Equi}
		f_Z(z;\mu,\lambda)
		= \sqrt{\frac{\lambda}{2\pi}}\, z^{-3/2}
		\exp\!\left(-\frac{\lambda (z-\mu)^2}{2\mu^2 z}\right),
		\quad z>0 .
	\end{align}
	Next, we set $\mu = d/v$ and $\lambda = d^2 / \sigma^2$ and substitute them into \eqref{Eq.Inv_Gau_PDF_Equi}. Observe that the prefactor is simplified to
	\[
	\sqrt{\frac{\lambda}{2\pi}}
	= \sqrt{\frac{d^2}{2\pi\sigma^2}}
	= \frac{d}{\sigma\sqrt{2\pi}}.
	\]
	Thus,
	\[
	\sqrt{\frac{\lambda}{2\pi}} z^{-3/2}
	= \frac{d}{\sigma\sqrt{2\pi}} z^{-3/2}.
	\]
	Next, we address the exponent $\lambda (z-\mu)^2 / 2\mu^2 z.$ Now, substituting $\lambda$ and $\mu$ gives
	\begin{align}
		\frac{\lambda (z-\mu)^2}{2\mu^2 z} = \frac{(d^2/\sigma^2)(z - d/v)^2}{2(d^2/v^2)z}
		= \frac{v^2 (z - d/v)^2}{2\sigma^2 z} = \frac{(vz - d)^2}{2\sigma^2 z}.
	\end{align}
	Hence, we obtain
	\[
	\frac{\lambda (z-\mu)^2}{2\mu^2 z}
	= \frac{(d - vz)^2}{2\sigma^2 z}.
	\]
	Now, to derive the final form, we substitute both the prefactor and the exponent parts and obtain the first-passage time density as follows:
	\begin{align}
		f_Z(z) = \frac{d}{\sigma\sqrt{2\pi}}\, z^{-3/2}
		\exp\!\left(-\frac{(d - vz)^2}{2\sigma^2 z} \right) ,
	\end{align}
	which features an exponentially decaying tail with finite mean and variance. Observe that the zero drift first-passage time density reads
	\begin{align}
		\label{Eq.FPT_Zero_Drift}
		f_Z(z) = \frac{d}{\sigma\sqrt{2\pi}}\, z^{-3/2} \exp\!\left(-\frac{d^2}{2\sigma^2 z} \right),
	\end{align}
	which features a heavy tail with \emph{infinite} mean and variance.

	\section{Proof of Lemma \ref{Lem.Neg_Log_Vol}; Negative Log-Volume of Hypersphere}
	\label{App.Neg_Log_Vol}

	\begin{proof}
		Recall that
		\[
		\operatorname{Vol}\big(\S_{\fc_1}(n,r_0)\big)
		=
		\frac{(r_0\sqrt{\pi})^n}{\Gamma(1 + \frac n2)}.
		\]
		By Stirling's expansion \cite{Olver2010} for the Gamma function,
		\begin{align}
			\log\Gamma(x+1)
			&
			=
			\left(x+\frac12\right)\log x-x
			+\frac12\log(2\pi)
			+\frac{1}{12x}
			+O(x^{-3}),
		\end{align}
		for $x\to\infty.$ Setting \(x=n/2\), we obtain
		\[
		\log\Gamma\left(\frac n2+1\right)
		=
		\frac n2\log\frac n2-\frac n2
		+\frac12\log(\pi n)
		+\frac{1}{6n}
		+O(n^{-3}).
		\]
		Consequently,
		\[
		\begin{aligned}
			\log\operatorname{Vol}\big(\S_{\fc_1}(n,r_0)\big)
			&
			=
			n\log(r_0\sqrt{\pi})
			-\log\Gamma\left(\frac n2+1\right)\\
			&=
			n\log(r_0\sqrt{\pi})
			\hspace{-.5mm}-\frac n2\log\frac n2
			\hspace{-.5mm}+\frac n2
			\hspace{-.5mm}-\frac12\log(\pi n)
			\hspace{-.5mm}-\frac{1}{6n}
			\hspace{-.5mm}+O(n^{-3})\\
			& =
			\frac n2\log\left(\frac{2\pi e r_0^2}{n}\right)
			-\frac12\log(\pi n)
			+O(n^{-1}) .
		\end{aligned}
		\]
		Therefore,
		\begin{align}
			-\log\operatorname{Vol}\big(\S_{\fc_1}(n,r_0)\big)
			&
			=
			-\frac n2
			\log\left(
			\frac{2\pi e r_0^2}{n}
			\right)
			+\frac12\log(\pi n)
			+O(n^{-1})
			\nonumber\\[2mm]
			&=
			\frac n2
			\log\left(
			\frac{n}{2\pi e r_0^2}
			\right)
			+\frac12\log(\pi n)
			+O(n^{-1})
			\nonumber\\[2mm]
			&=
			\frac n2\log n
			-\frac n2\log 2
			-n\log r_0
			-\frac n2\log(\pi e)
			+o(n)
			\nonumber\\[2mm]
			&=
			\frac n2\log\frac n2
			-n\log r_0
			-\frac n2\log(\pi e)
			+o(n)
		\end{align}
		since $$\frac12\log(\pi n)=o(n)
		\qquad\text{and}\qquad
		O(n^{-1})=o(n).$$
		This completes the proof of Lemma \ref{Lem.Neg_Log_Vol}.
		
	\end{proof}

	\section{Proof of Lemma \ref{Lem.MGF}; Moment Generating Function of an IG RV}
	\label{App.MGF}
	
	In the following, we upper-bound the fourth centered moment $\mathbb{E}[\bar{Z}^4]$ of $Z\sim\mathrm{IG}(\mu,\lambda)$.

	\begin{proof}
		The moment-generating function for $Z\sim\mathrm{IG}(\mu,\lambda)$ is given by
		\cite[Ch. 2]{Chhikara24}
		\begin{align}
			G_Z(\alpha)
			=\mathbb{E}[e^{\alpha Z}]
			=\int_0^\infty e^{\alpha z}f_Z(z)\,dz
			=
			\exp\!\left[
			\frac{\lambda}{\mu}
			\left(
			1-\sqrt{1-\frac{2\mu^2\alpha}{\lambda}}
			\right)
			\right],
		\end{align}
		for every $\alpha<\lambda/(2\mu^2)$. Let
		\begin{align}
			S=1-\frac{2\mu^2\alpha}{\lambda}.
		\end{align}
		Then
		\begin{align}
			G_Z(\alpha)
			=
			\exp\!\left(
			\frac{\lambda}{\mu}(1-\sqrt{S})
			\right).
		\end{align}
		
		Since the first four non-central moments are obtained by evaluating
		the corresponding derivatives of the moment-generating function at
		$\alpha=0$, we first obtain
		\begin{align}
			\mathbb{E}[Z]
			=G_Z^{(1)}(0)
			& =\frac{d}{d\alpha}G_Z(\alpha)\Big|_{\alpha=0}
			\nonumber\\
			&=
			\left(
			\frac{\mu}{S^{1/2}}
			\right)
			\exp\!\left(
			\frac{\lambda}{\mu}(1-\sqrt{S})
			\right)
			\Bigg|_{\alpha=0}
			=\mu.
		\end{align}
		
		Similarly, the second non-central moment is
		\begin{align}
			\mathbb{E}[Z^2]
			=G_Z^{(2)}(0)
			& =\frac{d^2}{d\alpha^2}G_Z(\alpha)\Big|_{\alpha=0}
			\nonumber\\
			&=
			\left(
			\frac{\mu^3}{\lambda S^{3/2}}
			+\frac{\mu^2}{S}
			\right)
			\exp\!\left(
			\frac{\lambda}{\mu}(1-\sqrt{S})
			\right)
			\Bigg|_{\alpha=0}
			\nonumber\\
			&=
			\mu^2+\frac{\mu^3}{\lambda}.
		\end{align}
		
		For the third non-central moment, we have
		\begin{align}
			\mathbb{E}[Z^3]
			=G_Z^{(3)}(0)
			& =\frac{d^3}{d\alpha^3}G_Z(\alpha)\Big|_{\alpha=0}
			\nonumber\\
			& =
			 \left(
			\frac{3\mu^5}{\lambda^2S^{5/2}}
			+\frac{3\mu^4}{\lambda S^2}
			+\frac{3\mu^4}{\lambda S^{5/2}}
			+\frac{\mu^3}{S^{3/2}}
			\right)
			\exp\!\left(
			\frac{\lambda}{\mu}(1-\sqrt{S})
			\right)
			\Bigg|_{\alpha=0}
			\nonumber\\
			&=
			\mu^3+\frac{3\mu^4}{\lambda}
			+\frac{3\mu^5}{\lambda^2}.
		\end{align}
		
		Finally, for the fourth non-central moment,
		\begin{align}
			\mathbb{E}[Z^4]
			=G_Z^{(4)}(0)
			& = \frac{d^4}{d\alpha^4}G_Z(\alpha)\Big|_{\alpha=0}
			\nonumber\\&
			=
			\left(
			\frac{15\mu^7}{\lambda^3S^{7/2}}
			+\frac{15\mu^6}{\lambda^2S^3}
			+\frac{6\mu^5}{\lambda S^{5/2}}
			+\frac{\mu^4}{S^2}
			\right)
			\exp\!\left(
			\frac{\lambda}{\mu}(1-\sqrt{S})
			\right)
			\Bigg|_{\alpha=0}
			\nonumber\\
			&=
			\mu^4+\frac{6\mu^5}{\lambda}
			+\frac{15\mu^6}{\lambda^2}
			+\frac{15\mu^7}{\lambda^3}.
			\label{Eq.Z4}
		\end{align}
		
		Using the binomial expansion, the fourth centered moment can be written as
		\begin{align}
			\mathbb{E}[\bar{Z}^4]
			&=
			\mathbb{E}[Z^4]
			-4\mu\mathbb{E}[Z^3]
			+6\mu^2\mathbb{E}[Z^2]
			-4\mu^3\mathbb{E}[Z]
			+\mu^4.
		\end{align}
		Now, let us define
		\begin{align}
			A_1(\mu)
			&\triangleq
			\mu^4+\frac{6\mu^5}{\lambda}
			+\frac{15\mu^6}{\lambda^2}
			+\frac{15\mu^7}{\lambda^3},\\
			A_2(\mu)
			&\triangleq
			-4\mu\left(
			\mu^3+\frac{3\mu^4}{\lambda}
			+\frac{3\mu^5}{\lambda^2}
			\right),\\
			A_3(\mu)
			&\triangleq
			6\mu^2\left(
			\mu^2+\frac{\mu^3}{\lambda}
			\right).
		\end{align}
		Then, substituting the above moments gives
		\begin{align}
			\mathbb{E}[\bar{Z}^4] = A_1(\mu)+A_2(\mu)+A_3(\mu)-4\mu^4+\mu^4 = 3\frac{\mu^6}{\lambda^2}
			+15\frac{\mu^7}{\lambda^3} = \frac{\mu^6}{\lambda^2}
			\left(3+15\frac{\mu}{\lambda}\right).
		\end{align}
		Since $\mu,\lambda>0$,
		\begin{align}
			3+15\frac{\mu}{\lambda}
			&\leq
			15\left(1+\frac{\mu}{\lambda}\right),
		\end{align}
		therefore,
		\begin{align}
			\mathbb{E}[\bar{Z}^4]
			\leq
			15\frac{\mu^6}{\lambda^2}
			\left(1+\frac{\mu}{\lambda}\right).
		\end{align}
	\end{proof}

	\section{Proof of Lemma \ref{Lem.AS_LB}; Almost-Sure Lower Bound for the Minimum of an i.i.d. IG Sequence}
	\label{App.AS_LB}

	\begin{proof}
		The proof proceeds in the following steps:
		\begin{itemize}
			\item Derive an upper bound on the lower tail probability
			$\Pr(Z_t\le\varepsilon)$ using the IG density.
			\item Apply the union bound to obtain a corresponding bound
			for the sample minimum $Z_{\min}(n)$.
			\item Choose a suitable threshold $\varepsilon_n$ so that
			the resulting probability bounds are summable.
			\item Apply the first Borel--Cantelli lemma to establish an
			almost-sure lower bound for $Z_{\min}(n)$ and deduce the
			corresponding asymptotic bounds for its reciprocal.
		\end{itemize}
		
		Let $\{Z_t\}_{t\ge1}$ be an i.i.d. distributed sequence of RVs with
		$Z_t\sim\mathrm{IG}(\mu,\lambda)$, where $\mu,\lambda>0$. Define the sample minimum for $n \geq 1$ by
		\begin{equation}
			Z_{\min}(n)\triangleq\min_{1\le t\le n}Z_t .
		\end{equation}
		The probability density function of $Z_t$ is
		\begin{equation}
			f_Z(z)
			=
			\sqrt{\frac{\lambda}{2\pi z^3}}
			\exp\left(
			-\frac{\lambda(z-\mu)^2}{2\mu^2z}
			\right),\qquad z>0.
		\end{equation}
		Expanding the exponent gives
		\begin{equation}
			\frac{\lambda(z-\mu)^2}{2\mu^2z}
			=
			\frac{\lambda z}{2\mu^2}
			-\frac{\lambda}{\mu}
			+\frac{\lambda}{2z}.
		\end{equation}
		Consequently,
		\begin{equation}
			f_Z(z) =
			\sqrt{\frac{\lambda}{2\pi}}e^{\lambda/\mu}
			z^{-3/2}
			\exp\left(
			-\frac{\lambda z}{2\mu^2}
			-\frac{\lambda}{2z}
			\right).
		\end{equation}
		Since $\exp(-\lambda z/(2\mu^2))\le1$, we obtain
		\begin{equation}
			f_Z(z)
			\le
			\sqrt{\frac{\lambda}{2\pi}}e^{\lambda/\mu}
			z^{-3/2}e^{-\lambda/(2z)}.
		\end{equation}
		For any $\varepsilon>0$, it follows that
		\begin{equation}
			\Pr(Z_t\le\varepsilon)
			\le
			\sqrt{\frac{\lambda}{2\pi}}e^{\lambda/\mu}
			\int_0^\varepsilon
			z^{-3/2}e^{-\lambda/(2z)}\,dz.
		\end{equation}
		Using the substitution $u=\lambda/(2z)$, we obtain
		\begin{equation}
			\Pr(Z_t\le\varepsilon)
			\le
			\frac{e^{\lambda/\mu}}{\sqrt{\pi}}
			\int_{\lambda/(2\varepsilon)}^\infty
			u^{-1/2}e^{-u}\,du.
		\end{equation}
		Let $a=\lambda/(2\varepsilon)$. Since
		$u^{-1/2}\le a^{-1/2}$ for $u\ge a$, we have
		\begin{align}
			\Pr(Z_t\le\varepsilon)
			\le
			\frac{e^{\lambda/\mu}}{\sqrt{\pi}}
			a^{-1/2}e^{-a}
			=
			\frac{e^{\lambda/\mu}}{\sqrt{\pi}}
			\sqrt{\frac{2\varepsilon}{\lambda}}
			e^{-\lambda/(2\varepsilon)}.
		\end{align}
		Next, define the constant
		\begin{equation}
			C\triangleq
			\frac{e^{\lambda/\mu}}{\sqrt{\pi}}
			\sqrt{\frac{2}{\lambda}}.
		\end{equation}
		Then, for $0<\varepsilon\le1$, we obtain
		\begin{align}
			\Pr(Z_t\le\varepsilon)
			\le
			C\sqrt{\varepsilon}\,
			e^{-\lambda/(2\varepsilon)}
			\le
			Ce^{-\lambda/(2\varepsilon)}.
		\end{align}
		The constant $C>0$ depends only on $\mu$ and $\lambda$ and is independent of $n$ and $\varepsilon$. Next, by the union bound and the definition of $Z_{\min}(n)$,
		\begin{align}
			\Pr(Z_{\min}(n)\le\varepsilon)
			=
			\Pr\left(
			\bigcup_{t=1}^n\{Z_t\le\varepsilon\}
			\right)
			\le
			\sum_{t=1}^n\Pr(Z_t\le\varepsilon)
			\le
			Cn e^{-\lambda/(2\varepsilon)},
		\end{align}
		for $0<\varepsilon\le1$. Fix any $\eta>1$ and choose the threshold
		\begin{equation}
			\varepsilon_n
			\triangleq
			\frac{\lambda}{2(1+\eta)\ln n},
			\qquad n\ge2.
		\end{equation}
		
		Since $\varepsilon_n\to0$, we have $\varepsilon_n\le1$
		for all sufficiently large $n$. Substituting $\varepsilon_n$
		into (2) yields
		\begin{align}
			\label{Ineq.Pr_Z_min}
			\Pr(Z_{\min}(n)\le\varepsilon_n)
			\le
			Cn e^{-(1+\eta)\ln n}
			=Cn^{-\eta}
		\end{align}
		for all sufficiently large $n$. Since $\eta>1$, the series
		\begin{equation}
			\sum_{n=2}^{\infty}
			\Pr(Z_{\min}(n)\le\varepsilon_n)
		\end{equation}
		converges. Hence, by the first Borel--Cantelli lemma,
		\begin{equation}
			\Pr\left(
			Z_{\min}(n)\le\varepsilon_n
			\text{ infinitely often}
			\right)=0.
		\end{equation}
		Therefore, with probability one, there exists a finite
		threshold beyond which
		\[
		Z_{\min}(n)>
		\frac{\lambda}{2(1+\eta)\ln n}.
		\]
		It follows that
		\begin{equation}
			Z_{\min}(n)
			=\Omega\left(\frac{1}{\ln n}\right)
			\quad\text{almost surely}.
		\end{equation}
		Taking reciprocals gives
		\begin{equation}
			\frac{1}{Z_{\min}(n)}=O(\ln n),
			\qquad
			\frac{1}{Z_{\min}^2(n)}=O(\ln^2 n)
			\quad\text{almost surely}.
		\end{equation}
		Finally, we clarify the meaning of the threshold in the
		almost-sure statement. Define the good event
		\begin{equation}
			\mathcal{G}
			\triangleq
			\left\{
			Z_{\min}(n)\ge
			\frac{\lambda}{2(1+\eta)\ln n}
			\right\},
		\end{equation}
		and let $\mathcal{G}^c$ denote its complement. The
		first Borel--Cantelli lemma implies
		\[
		\Pr\left(
		\mathcal{G}^c
		\text{ infinitely often}
		\right)=0.
		\]
		Thus, with probability one, the event $\mathcal{G}^c$
		occurs only finitely many times. Equivalently, the lower
		bound holds for all sufficiently large $n$, almost surely.
		The starting index from which the bound holds may depend
		on the realized sample sequence and need not be deterministic.
		This dependence is inherent in the almost-sure statement and
		does not affect the claimed asymptotic orders. This completes
		the proof.
	\end{proof}

	\section{Proof of Lemma \ref{Lem.Log_Likelihood_Ratio}; Output-Density Ratio Bound}
	\label{App.Log_likelihood_Ratio}
	In the following, we provide the proof of Lemma \ref{Lem.Log_Likelihood_Ratio}.
	To establish the desired output-density ratio bound, we adopt a logarithmic
	perturbation approach and proceed as follows:
	
	\begin{itemize}[leftmargin=*]
		\item We decompose the ratio of the two output densities into a polynomial prefactor and an exponential term.
		\item We take the logarithm of the density ratio to transform the multiplicative structure into an additive one, allowing the two contributions to be analyzed separately.
		\item We introduce suitable auxiliary quantities to characterize the perturbations induced by the differences between the codewords.
		\item We exploit the imposed codeword constraints and the almost-sure lower bound on the noise samples in Lemma \ref{Lem.AS_LB} to derive uniform asymptotic bounds on the resulting terms.
		\item We combine these estimates to establish a two-sided bound on the log-likelihood ratio.
		\item Finally, we exponentiate the resulting bound to obtain the desired output-density ratio bound.
	\end{itemize}

	
	For a fixed reference codeword $\fc_{i_1}$, define the following observation-space set:
	\begin{align}
		\mathcal{G}_{i_1}
		&\triangleq
		\left\{
		\fy\in\mathbb{R}^n:
		\min_{1\le t\le n}(y_t-c_{i_1,t})\ge\varepsilon_n
		\right\},
		\\
		\varepsilon_n
		&\triangleq
		\frac{\lambda}{2(1+\eta)\ln n},
		\qquad \eta>1.
	\end{align}
	Without loss of generality, we take $\fc_{i_1}$ as the transmitted codeword and $\fc_{i_2}$ as the competing codeword. When $\fc_{i_1}$ is transmitted, the corresponding channel output is
	\[
	\mathbf Y(i_1)=\fc_{i_1}+\mathbf Z.
	\]
	Hence, by Lemma~\ref{Lem.AS_LB}, with probability one, $\mathbf Y(i_1)\in\mathcal G_{i_1}$ for all sufficiently large $n$. The index beyond which this holds may depend on the noise realization. We now consider any pair $(\mathbf c_{i_1},\mathbf c_{i_2})\in\mathcal P_n$, for which $|c_{i_1,t}-c_{i_2,t}| \le 2\alpha_n, t\in[\![n]\!],$ where \(\alpha_n=a^2/n^{1+2b}\), with \(a,b>0\). Since \(\alpha_n=o(\varepsilon_n)\), for all sufficiently large \(n\) and every \(\fy\in\mathcal{G}_{i_1}\), we have
	\begin{align}
		y_t-c_{i_1,t}
		&\ge \varepsilon_n,
		\\
		y_t-c_{i_2,t}
		&= (y_t-c_{i_1,t})+(c_{i_1,t}-c_{i_2,t})
		\nonumber\\
		&\ge (y_t-c_{i_1,t})-|c_{i_1,t}-c_{i_2,t}|
		\nonumber\\
		&\ge \varepsilon_n-2\alpha_n
		\ge \frac{\varepsilon_n}{2},
	\end{align}
	where the last inequality follows from \(\alpha_n=o(\varepsilon_n)\). Therefore, we can define the following conventions for every \(t\in[\![n]\!]\)
	\begin{align}
		\delta_{i_1,i_2,t}^A
		&\triangleq
		\frac{c_{i_1,t}-c_{i_2,t}}{y_t-c_{i_2,t}},
		\\
		\delta_{i_1,i_2,t}^B
		&\triangleq
		\frac{c_{i_2,t}-c_{i_1,t}}{(y_t-c_{i_1,t})(y_t-c_{i_2,t})}.
	\end{align}
	On \(\mathcal{G}_{i_1}\), these quantities satisfy the following uniform bounds
	\begin{align}
		\left|\delta_{i_1,i_2,t}^A\right|
		&=
		\left|
		\frac{c_{i_1,t}-c_{i_2,t}}{y_t-c_{i_2,t}}
		\right|
		\le \frac{4\alpha_n}{\varepsilon_n},
		\\
		\left|\delta_{i_1,i_2,t}^B\right|
		&=
		\left|
		\frac{c_{i_2,t}-c_{i_1,t}}{(y_t-c_{i_1,t})(y_t-c_{i_2,t})}
		\right|
		\le \frac{4\alpha_n}{\varepsilon_n^2}.
	\end{align}
	These bounds hold uniformly over all \(\fy\in \mathcal{G}_{i_1}\), all \(t\in[\![n]\!]\), and all codeword pairs satisfying the stated coordinatewise constraint.

	\begin{proof}
		We start by calculating the logarithm of the ratio of noise PDFs corresponding to two distinct codewords \(\fc_{i_1}\) and \(\fc_{i_2}\), i.e.,
		\[
		\frac{f_{\fZ}(\fy-\fc_{i_2})}{f_{\fZ}(\fy-\fc_{i_1})}.
		\]
		Recall that the noise density provided in \eqref{Eq.Noise_PDF_Vect} is given by
		\begin{align}
			f_{\fZ}(\fz)
			&=\prod_{t=1}^n f_{Z_t}(z_t)
			=\Big(\frac{\lambda}{2\pi}\Big)^{n/2}
			\prod_{t=1}^n z_t^{-3/2}
			\exp\left(
			-\frac{\lambda}{2\mu^2}
			\sum_{t=1}^n\frac{(z_t-\mu)^2}{z_t}
			\right).
		\end{align}
		Consequently,
		\[
		\frac{f_{\fZ}(\fy-\fc_{i_2})}
		{f_{\fZ}(\fy-\fc_{i_1})}
		=A_{i_1,i_2}\cdot B_{i_1,i_2},
		\]
		where
		\begin{align}
			A_{i_1,i_2}
			&=\prod_{t=1}^n
			\left(\frac{y_t-c_{i_1,t}}{y_t-c_{i_2,t}}\right)^{3/2},
			\\
			B_{i_1,i_2}
			&=\exp\left(
			-\frac{\lambda\upDelta_{i_1,i_2}}{2\mu^2}
			\right),
		\end{align}
		with
		\begin{align}
			\upDelta_{i_1,i_2}
			\triangleq
			\sum_{t=1}^n
			\left[
			\frac{(y_t-c_{i_2,t}-\mu)^2}{y_t-c_{i_2,t}}
			-
			\frac{(y_t-c_{i_1,t}-\mu)^2}{y_t-c_{i_1,t}}
			\right].
		\end{align}
		Therefore, the logarithm of the ratio of the noise densities reads
		\begin{align}
			\label{Eq.Log_Ratio}
			\ln\left(
			\frac{f_{\fZ}(\fy-\fc_{i_2})}{f_{\fZ}(\fy-\fc_{i_1})}
			\right)
			=\ln A_{i_1,i_2}+\ln B_{i_1,i_2}.
		\end{align}
		For the first term, we obtain
		\begin{align}
			\ln A_{i_1,i_2}
			&=\frac32\sum_{t=1}^n
			\ln\left(
			\frac{y_t-c_{i_1,t}}{y_t-c_{i_2,t}}
			\right)
			\nonumber\\
			&=\frac32\sum_{t=1}^n
			\ln\left(
			1-\frac{c_{i_1,t}-c_{i_2,t}}{y_t-c_{i_2,t}}
			\right)
			\nonumber\\
			&=\frac32\sum_{t=1}^n
			\ln(1-\delta_{i_1,i_2,t}^A).
		\end{align}

		For the second term, recall that, when \(\fc_{i_1}\) is transmitted, \(\fY=\fc_{i_1}+\fZ\). Accordingly, define the actual noise coordinate and the codeword difference as
		\begin{align}
			z_t &\triangleq y_t-c_{i_1,t},\\
			d_t &\triangleq c_{i_1,t}-c_{i_2,t}.
		\end{align}
		Then
		\[
		y_t-c_{i_2,t}=z_t+d_t.
		\]
		On the event \(\mathcal{G}_{i_1}\), we have \(z_t\ge\varepsilon_n\). Moreover, since \(|d_t|\le 2\alpha_n\) and \(\alpha_n=o(\varepsilon_n)\), for all sufficiently large \(n\),
		\[
		z_t+d_t=y_t-c_{i_2,t}\ge\frac{\varepsilon_n}{2}>0.
		\]
		
		Using the identity
		\[
		\frac{(z-\mu)^2}{z}
		=z-2\mu+\frac{\mu^2}{z},
		\qquad z>0,
		\]
		we can express the difference of the quadratic terms as
		\begin{align}
			&\frac{(y_t-c_{i_2,t}-\mu)^2}{y_t-c_{i_2,t}}
			-\frac{(y_t-c_{i_1,t}-\mu)^2}{y_t-c_{i_1,t}}
			\nonumber\\
			&\quad=
			\frac{(z_t+d_t-\mu)^2}{z_t+d_t}
			-\frac{(z_t-\mu)^2}{z_t}
			\nonumber\\
			&\quad=
			d_t+\mu^2\left(\frac{1}{z_t+d_t}-\frac{1}{z_t}\right)
			\nonumber\\
			&\quad=
			d_t-\frac{\mu^2d_t}{z_t(z_t+d_t)}.
		\end{align}
		
		Consequently, recalling that \(\fY=\fc_{i_1}+\fZ\) when \(\fc_{i_1}\) is transmitted, let
		\[
		z_t\triangleq y_t-c_{i_1,t},
		\qquad
		d_t\triangleq c_{i_1,t}-c_{i_2,t}.
		\]
		Then \(y_t-c_{i_2,t}=z_t+d_t\), and we obtain
		\begin{align}
			\label{Eq.Log_B}
			\ln B_{i_1,i_2}
			&=-\frac{\lambda}{2\mu^2}
			\sum_{t=1}^n
			\left[
			\frac{(z_t+d_t-\mu)^2}{z_t+d_t}
			-\frac{(z_t-\mu)^2}{z_t}
			\right]
			\nonumber\\
			&=-\frac{\lambda}{2\mu^2}
			\sum_{t=1}^n
			\left[
			d_t+\mu^2\left(
			\frac{1}{z_t+d_t}-\frac{1}{z_t}
			\right)
			\right]
			\nonumber\\
			&=-\frac{\lambda}{2\mu^2}
			\sum_{t=1}^n
			\left[
			d_t-\frac{\mu^2d_t}{z_t(z_t+d_t)}
			\right]
			\nonumber\\
			&=-\frac{\lambda}{2\mu^2}
			\sum_{t=1}^n
			\left[
			(c_{i_1,t}-c_{i_2,t})
			+\mu^2\delta_{i_1,i_2,t}^B
			\right],
		\end{align}
		where
		\begin{align}
			\delta_{i_1,i_2,t}^B
			\triangleq
			\frac{c_{i_2,t}-c_{i_1,t}}{(y_t-c_{i_1,t})(y_t-c_{i_2,t})}
			=-\frac{d_t}{z_t(z_t+d_t)}.
		\end{align}
		The second equality in \eqref{Eq.Log_B} follows by expanding
		\[
		\frac{(z-\mu)^2}{z}
		=z-2\mu+\frac{\mu^2}{z},
		\]
		and the third equality follows from
		\[
		\frac{1}{z_t+d_t}-\frac{1}{z_t}
		=-\frac{d_t}{z_t(z_t+d_t)}.
		\]

		Next, to incorporate the effect of the upper bound in \eqref{Eq.alpha_n}, we simplify some expressions in \eqref{Eq.Log_Ratio}. All bounds below are established for \(\fy \in \mathcal{G}_{i_1}\).
		
		\begin{align}
			\label{Ineq.A_i1_i2_UB}
			|\ln A_{i_1,i_2}|
			&\stackrel{(a)}{\le}
			\frac32\sum_{t=1}^n
			|\ln(1-\delta_{i_1,i_2,t}^A)|
			\nonumber\\
			&\stackrel{(b)}{\le}
			\frac32\sum_{t=1}^n
			\left(
			|\delta_{i_1,i_2,t}^A|
			+\mathcal O\big((\delta_{i_1,i_2,t}^A)^2\big)
			\right)
			\nonumber\\
			&\stackrel{(c)}{\le}
			\frac32\sum_{t=1}^n|\delta_{i_1,i_2,t}^A|
			+\mathcal O\left(
			\sum_{t=1}^n(\delta_{i_1,i_2,t}^A)^2
			\right)
			\nonumber\\
			&\stackrel{(d)}{\le}
			\frac32\sum_{t=1}^n\frac{4\alpha_n}{\varepsilon_n}
			+\mathcal O\left(
			\sum_{t=1}^n\frac{16\alpha_n^2}{\varepsilon_n^2}
			\right)
			\nonumber\\
			&=
			\mathcal O\left(\frac{n\alpha_n}{\varepsilon_n}\right)
			+\mathcal O\left(\frac{n\alpha_n^2}{\varepsilon_n^2}\right)
			\nonumber\\
			&=
			\mathcal O\left(\frac{\ln n}{n^{2b}}\right)
			+\mathcal O\left(\frac{\ln^2 n}{n^{1+4b}}\right)
			\nonumber\\
			&=
			\mathcal O\left(\frac{\ln n}{n^{2b}}\right)
			\triangleq \tau_n^{B}.
		\end{align}
		where
		\begin{itemize}[leftmargin=*]
			\item \((a)\) follows from the triangle inequality and the expression for \(\ln A_{i_1,i_2}\).
			
			\item \((b)\) follows from the Taylor expansion
			\[
			\ln(1-\delta)
			=-\delta-\frac{\delta^2}{2}
			+\mathcal O(\delta^3),
			\]
			which yields
			\[
			|\ln(1-\delta)|
			\le |\delta|+C\delta^2,
			\qquad |\delta|\le\frac12,
			\]
			for some absolute constant \(C>0\). This bound applies uniformly for all sufficiently large \(n\), since
			\[
			|\delta_{i_1,i_2,t}^A|
			\le\frac{4\alpha_n}{\varepsilon_n}
			=\mathcal O\left(\frac{\ln n}{n^{1+2b}}\right)
			\longrightarrow0.
			\]
			
			\item \((c)\) separates the first-order and second-order contributions, with the constant in the quadratic remainder absorbed into the \(\mathcal O(\cdot)\) notation.
			
			\item \((d)\) uses the uniform bounds on \( \mathcal{G}_{i_1}\):
			\[
			|\delta_{i_1,i_2,t}^A|
			\le\frac{4\alpha_n}{\varepsilon_n},
			\qquad
			(\delta_{i_1,i_2,t}^A)^2
			\le\frac{16\alpha_n^2}{\varepsilon_n^2}.
			\]
			
			\item The final steps use
			\[
			\varepsilon_n^{-1}=\mathcal O(\ln n),
			\qquad
			\alpha_n=\frac{a^2}{n^{1+2b}},
			\]
			and the fact that, since \(b>0\),
			\[
			\frac{\ln^2 n}{n^{1+4b}}
			=o\left(\frac{\ln n}{n^{2b}}\right).
			\]
		\end{itemize}

		Now, we focus on $\ln B_{i_1,i_2}$. Observe that
		\begin{align}
			\label{Ineq.B_i1_i2_UB}
			|\ln B_{i_1,i_2}|
			&\stackrel{(a)}{\le}
			\frac{\lambda}{2\mu^2}
			\sum_{t=1}^n
			\left(
			|c_{i_1,t}-c_{i_2,t}|
			+\mu^2|\delta_{i_1,i_2,t}^B|
			\right)
			\nonumber\\
			&\stackrel{(b)}{\le}
			\frac{\lambda}{2\mu^2}
			\left(
			2n\alpha_n
			+\frac{4\mu^2n\alpha_n}{\varepsilon_n^2}
			\right)
			\nonumber\\
			&=
			\mathcal{O}(n\alpha_n)
			+\mathcal{O}\left(\frac{n\alpha_n}{\varepsilon_n^2}\right)
			\nonumber\\
			&=
			\mathcal{O}\left(\frac{1}{n^{2b}}\right)
			+\mathcal{O}\left(\frac{\ln^2 n}{n^{2b}}\right)
			\nonumber\\
			&=
			\mathcal{O}\left(\frac{\ln^2 n}{n^{2b}}\right)
			\triangleq \tau_n^{A},
		\end{align}
		where
		\begin{itemize}[leftmargin=*]
			\item $(a)$ uses the generalized triangle inequality for \eqref{Eq.Log_B}.
			\item $(b)$ uses \eqref{Eq.alpha_n} and the uniform bound on $\mathcal{G}_{i_1}$:
			\[
			|\delta_{i_1,i_2,t}^B|
			\le\frac{4\alpha_n}{\varepsilon_n^2}.
			\]
			\item The remaining steps use
			$\varepsilon_n^{-2}=\mathcal{O}(\ln^2 n)$ and
			$n\alpha_n=\mathcal{O}(n^{-2b})$.
		\end{itemize}

		Thereby, combining \eqref{Ineq.A_i1_i2_UB} and \eqref{Ineq.B_i1_i2_UB} with \eqref{Eq.Log_Ratio}, we obtain
		\begin{align}
			\ln\left(
			\frac{f_{\fZ}(\fy-\fc_{i_2})}{f_{\fZ}(\fy-\fc_{i_1})}
			\right)
			&=\ln A_{i_1,i_2}+\ln B_{i_1,i_2}
			\nonumber\\
			&\le
			|\ln A_{i_1,i_2}|+|\ln B_{i_1,i_2}|
			\nonumber\\
			&\le \tau_n^{A} + \tau_n^{B} .
		\end{align}
		Similarly, the absolute-value bounds imply
		\begin{align}
			\ln\left(
			\frac{f_{\fZ}(\fy-\fc_{i_2})}{f_{\fZ}(\fy-\fc_{i_1})}
			\right)
			&=\ln A_{i_1,i_2}+\ln B_{i_1,i_2}
			\nonumber\\
			&\ge
			-|\ln A_{i_1,i_2}|-|\ln B_{i_1,i_2}|
			\nonumber\\
			&\ge-(\tau_n^{A} + \tau_n^{B}).
		\end{align}
		Exponentiating these bounds yields
		\begin{align}
			\label{Eq.Log_Lik-Ratio_1}
			e^{-(\tau_n^{A} + \tau_n^{B})}
			\le
			\frac{f_{\fZ}(\fy-\fc_{i_2})}{f_{\fZ}(\fy-\fc_{i_1})}
			\le
			e^{\tau_n^{A} + \tau_n^{B}}.
		\end{align}
		Consequently,
		\begin{align}
			\sup_{\fy\in\mathcal G_{i_1}}
			\left|
			1-
			\frac{f_{\fZ}(\fy-\fc_{i_2})}{f_{\fZ}(\fy-\fc_{i_1})}
			\right|
			\le e^{\tau_n^{A} + \tau_n^{B}}-1.
		\end{align}
		Since $\tau_n^{A}+\tau_n^{B}\to 0$ as $n\to\infty$, the right-hand side converges to zero. Hence, the output-density ratio converges uniformly to one on $\mathcal G_{i_1}$, meaning that the ratio approaches one uniformly over all observations in the good set, rather than only for a particular observation. This uniform convergence is useful because, by Lemma~\ref{Lem.AS_LB}, the actual channel output $\mathbf Y(i_1)$ belongs to $\mathcal G_{i_1}$ almost surely for all sufficiently large $n$. This completes the proof.
	\end{proof}

	\section{Proof of Lemma~\ref{Lem.Converse}; Separation Property of Identification Codebooks}
	\label{App.Converse_Proof}

	We establish Lemma~\ref{Lem.Converse} via a proof by contradiction. To this end, suppose that the condition in \eqref{Ineq.Conv_Distance} is violated. Under this assumption, we derive an inconsistency by showing that the combined type~I and type~II error probabilities converge to one, that is, $$\lim_{n \to \infty} \big[ P_{e,1}(i_1) + P_{e,2}(i_2,i_1) \big] = 1.$$ which contradicts the assumed behavior, i.e., the reliability of the identification codes sequence.
	
	\begin{proof}
		Fix \(e_1,e_2>0\). Let \(\gamma_1,\gamma_2,\gamma_3,\zeta,\gamma>0\) be arbitrarily small, and choose \(0<\theta<\zeta/2\). Suppose, for a contradiction, that there exist distinct messages \(i_1\) and \(i_2\) such that, for all \(t\in[\![n]\!]\),
		\begin{align}
			\label{Eq.alpha_n}
			| c_{i_1,t} - c_{i_2,t} | < 2\alpha_n = \frac{2a^2}{n^{1+2b}} ,
		\end{align}
		where $a,b>0$ are fixed and arbitrary small constants, respectively. By the definition of \(\mathcal P_n\), this assumption implies that $(\fc_{i_1},\fc_{i_2})\in\mathcal P_n,$
		so Lemma~\ref{Lem.Log_Likelihood_Ratio} applies to this pair.
		
		Next, we introduce the following subsets of \(\mathbb{R}^n\):
		\begin{align}
			\mathbbmss{E}_{i_2}
			&\triangleq
			\left\{
			\fy\in\mathbb{R}^n:
			\left\|\fy-\fc_{i_2}-\mu\mathbf{1}\right\|
			\leq
			\sqrt{n\left(\frac{\mu^3}{\lambda}+\zeta\right)}
			\right\},
			\label{Eq.Event_BC}
			\\
			\mathbbmss D_{i_1,i_2} & \triangleq \mathbbmss D_{i_1}\cap\mathbbmss E_{i_2} = \left\{
			\fy\in\mathbbmss{D}_{i_1}:
			\left\|\fy-\fc_{i_2}-\mu\mathbf{1}\right\|
			\leq
			\sqrt{n\left(\frac{\mu^3}{\lambda}+\zeta\right)}
			\right\},
			\\
			\mathbbmss{H}_{i_1,i_2}
			&\triangleq
			\mathbbmss{D}_{i_1,i_2} \cap \mathcal{G}_{i_1}.
		\end{align}
		The set \(\mathbbmss{H}_{i_1,i_2}\) restricts the integration region to the subset of \(\mathbbmss{D}_{i_1,i_2}\) on which the density-ratio bound is valid. Here,
		\begin{equation}
			\mathcal{G}_{i_1}
			\triangleq
			\left\{
			\fy\in\mathbb{R}^n:
			\min_{1\leq t\leq n}(y_t-c_{i_1,t})
			\geq\varepsilon_n
			\right\},
			\qquad
			\varepsilon_n
			\triangleq
			\frac{\lambda}{2(1+\eta)\ln n},
			\quad \eta>1.
		\end{equation}

		Next, we bound the type-I error probability by decomposing the decoding region. Since \(\mathbbmss{H}_{i_1,i_2}\subseteq\mathbbmss{D}_{i_1}\), we can decompose \(\mathbbmss{D}_{i_1}\) into two disjoint sets:
		\begin{align}
			\mathbbmss{D}_{i_1}
			&=
			\mathbbmss{H}_{i_1,i_2}
			\mathbin{\dot\cup}
			\bigl(
			\mathbbmss{D}_{i_1}
			\setminus\mathbbmss{H}_{i_1,i_2}
			\bigr) ,
			\label{Eq.DecodingRegionDecomp}
		\end{align}
		where \(\mathbin{\dot\cup}\) denotes a disjoint union, i.e., the two sets have an empty intersection. Furthermore, by De Morgan's law and the distributive law,
		\begin{align}
			\mathbbmss{D}_{i_1}\setminus\mathbbmss{H}_{i_1,i_2}
			&=
			\mathbbmss{D}_{i_1}\cap
			\bigl(
			\mathbbmss{D}_{i_1,i_2}^{\,c}
			\cup\mathcal{G}_{i_1}^{\,c}
			\bigr)
			\nonumber\\
			&=
			\bigl(
			\mathbbmss{D}_{i_1}
			\setminus\mathbbmss{D}_{i_1,i_2}
			\bigr)
			\cup
			\bigl(
			\mathbbmss{D}_{i_1}
			\cap\mathcal{G}_{i_1}^{\,c}
			\bigr)
			\nonumber\\
			&\subseteq
			\bigl(
			\mathbbmss{D}_{i_1}
			\setminus\mathbbmss{D}_{i_1,i_2}
			\bigr)
			\cup\mathcal{G}_{i_1}^{\,c}.
			\label{Eq.DecodingRegionComplement}
		\end{align}
		By the definition
		\(\mathbbmss{D}_{i_1,i_2}
		=\mathbbmss{D}_{i_1}\cap\mathbbmss{E}_{i_2}\), we further have
		\begin{align}
			\mathbbmss{D}_{i_1}
			\setminus\mathbbmss{D}_{i_1,i_2}
			&=
			\mathbbmss{D}_{i_1}
			\setminus
			\bigl(
			\mathbbmss{D}_{i_1}\cap\mathbbmss{E}_{i_2}
			\bigr)
			\nonumber\\
			&=
			\mathbbmss{D}_{i_1}
			\cap\mathbbmss{E}_{i_2}^{\,c}
			\subseteq
			\mathbbmss{E}_{i_2}^{\,c}.
			\label{Eq.DecodingRegionDifference}
		\end{align}
		Combining \eqref{Eq.DecodingRegionComplement} and
		\eqref{Eq.DecodingRegionDifference} yields
		\begin{align}
			\mathbbmss{D}_{i_1}
			\setminus\mathbbmss{H}_{i_1,i_2}
			\subseteq
			\mathbbmss{E}_{i_2}^{\,c}
			\cup\mathcal{G}_{i_1}^{\,c}.
			\label{Eq.DecodingRegionComplementBound}
		\end{align}
		Consequently, using the disjoint decomposition in
		\eqref{Eq.DecodingRegionDecomp} and the nonnegativity of
		\(f_{\fZ}\), we obtain
		\begin{align}
			1-P_{e,1}(i_1)
			&=
			\int_{\mathbbmss{D}_{i_1}}
			f_{\fZ}(\fy-\fc_{i_1})\,d\fy
			\nonumber\\
			&=
			\int_{\mathbbmss{H}_{i_1,i_2}}
			f_{\fZ}(\fy-\fc_{i_1})\,d\fy
			+
			\int_{\mathbbmss{D}_{i_1}\setminus
			\mathbbmss{H}_{i_1,i_2}}
			f_{\fZ}(\fy-\fc_{i_1})\,d\fy
			\nonumber\\
			&\leq
			\int_{\mathbbmss{H}_{i_1,i_2}}
			f_{\fZ}(\fy-\fc_{i_1})\,d\fy
			+
			\int_{\mathbbmss{E}_{i_2}^{\,c}}
			f_{\fZ}(\fy-\fc_{i_1})\,d\fy
			+
			\int_{\mathcal{G}_{i_1}^{\,c}}
			f_{\fZ}(\fy-\fc_{i_1})\,d\fy.
			\label{Eq.Pe1boundConv0Fast}
		\end{align}

		We first bound the integral over \(\mathbbmss E_{i_2}^{\,c}\). For every \(\fy\in\mathbbmss E_{i_2}^{\,c}\), the reverse triangle inequality gives
		\begin{align}
			\label{Ineq.RTI_Conv}
			\| \fy - \fc_{i_1} - \mu \mathbf 1 \|
			&\stackrel{(a)}{\geq}
			\| \fy - \fc_{i_2} - \mu \mathbf 1 \|
			- \| \fc_{i_1} - \fc_{i_2} \|
			\nonumber\\
			&\stackrel{(b)}{>}
			\sqrt{n\left(\frac{\mu^3}{\lambda}+\zeta\right)}
			- \| \fc_{i_1} - \fc_{i_2} \|
			\nonumber\\
			&\stackrel{(c)}{\geq}
			\sqrt{n\left(\frac{\mu^3}{\lambda}+\zeta\right)}
			- 2\alpha_n\sqrt n
			\nonumber\\
			&\stackrel{(d)}{=}
			\sqrt{n\left(\frac{\mu^3}{\lambda}+\zeta\right)}
			-\frac{2a^2}{n^{\frac 12 +2b}}.
		\end{align}
		\begin{itemize}
			\item $(a)$ uses the reverse triangle inequality.
			
			\item $(b)$ follows from the assumed lower bound on $\|\fy-\fc_{i_2}-\mu\mathbf{1}\|$, since $\fy\in\mathbbmss{E}_{i_2}^c$.
			
			\item $(c)$ uses the definition of the Euclidean norm and the assumption that $|c_{i_1,t}-c_{i_2,t}|<2\alpha_n$ given in \eqref{Eq.alpha_n} for every $t \in [\![n]\!]$. In particular,
			\begin{align}
				\|\fc_{i_1}-\fc_{i_2}\|_2
				=\sqrt{\sum_{t=1}^{n}|c_{i_1,t}-c_{i_2,t}|^2}
				\leq \sqrt{\sum_{t=1}^{n}(2\alpha_n)^2}
				=\sqrt{n(2\alpha_n)^2}
				=2\alpha_n\sqrt{n}.
			\end{align}
			
			\item $(d)$ follows by substituting $\alpha_n=a^2/n^{1+2b}$, which gives
			\[
			2\alpha_n\sqrt{n}
			=\frac{2a^2}{n^{\frac 12 +2b}}.
			\]
		\end{itemize}

		Next, define the following events.
		\begin{align}
			r_\theta
			&\triangleq
			\sqrt{n\left(\frac{\mu^3}{\lambda}+\theta \right)},
			&
			r_\zeta
			&\triangleq
			\sqrt{n\left(\frac{\mu^3}{\lambda}+\zeta\right)},
		\end{align}
		and let
		\begin{align}
			\mathbbmss{F}_{i_1}^{\,c}
			\triangleq
			\left\{
			\fy\in\mathbb{R}^n:
			\|\fy-\fc_{i_1}-\mu\mathbf{1}\|>r_\theta
			\right\}.
			\label{Eq.Regiong0}
		\end{align}
		Since \(0<\theta<\zeta/2\), we have \(r_\theta<r_\zeta\). Moreover, the codeword-distance bound established above yields
		\[
		\|\fc_{i_1}-\fc_{i_2}\|
		\leq \frac{2a^2}{n^{\frac 12 +2b}}.
		\]
		Thus, for every \(\fy\in\mathbbmss{E}_{i_2}^{\,c}\), the reverse triangle inequality calculated in \eqref{Ineq.RTI_Conv} gives
		\begin{align}
			\|\fy-\fc_{i_1}-\mu\mathbf{1}\|
			&\geq
			\|\fy-\fc_{i_2}-\mu\mathbf{1}\|
			-\|\fc_{i_1}-\fc_{i_2}\|
			\nonumber\\
			&>
			r_\zeta-\frac{2a^2}{n^{\frac 12 +2b}}.
		\end{align}
		Since
		\[
		r_\zeta-r_\theta
		=
		\sqrt{n}\left(
		\sqrt{\frac{\mu^3}{\lambda}+\zeta}
		-
		\sqrt{\frac{\mu^3}{\lambda}+\theta}
		\right)
		\]
		grows as \(\sqrt n\), whereas \((2a^2)/n^{1+2b} \to 0\) for \(b>0\), we have, for all sufficiently large \(n\),
		\[
		r_\zeta-\frac{2a^2}{n^{\frac 12 +2b}}>r_\theta.
		\]
		Consequently,
		\begin{align}
			\mathbbmss{E}_{i_2}^{\,c}
			\subseteq
			\mathbbmss{F}_{i_1}^{\,c}.
			\label{Eq.EsubsetF}
		\end{align}
		
		Hence, the integral over \(\mathbbmss E_{i_2}^{\,c}\) in \eqref{Eq.Pe1boundConv0Fast} is bounded by
		\begin{align}
			\label{Ineq.F_i_1_i_2_compl}
			\int_{\mathbbmss E_{i_2}^{\,c}}
			f_{\fZ}(\fy-\fc_{i_1})\,d\fy \leq \int_{\mathbbmss{F}_{i_1}^c} \hspace{-1mm} f_{\fZ}( \fy - \fc_{i_1}) d\fy
			&
			= \Pr \left( \| \fY(i_1) - \fc_{i_1} - \mu \mathbf 1 \| \hspace{-.7mm} > \hspace{-.7mm} \sqrt{n\left( \frac{\mu^3}{\lambda} + \eta \right)} \right)
			\nonumber\\&
			\stackrel{(a)}{=} \hspace{-.7mm} \Pr \left( n^{-1} \| \bar{\fZ} \|^2 - \left( \frac{\mu^3}{\lambda} \right) > \eta \right)
			\nonumber\\&
			\stackrel{(b)}{\leq} \frac{15 \mu^6}{n\lambda^2\eta^2}
			\left(1+\frac{\mu}{\lambda}\right) \leq \gamma_1,
		\end{align}
		for sufficiently large $n$, where $(a)$ holds by the Chebyshev's inequality, followed by the substitution of
		$\fZ = \fY(i_1) - \fc_{i_1}$ and $(b)$ invokes \eqref{Ineq.Var_Decoding_Metric}
		\begin{align}
			\operatorname{Var}\!\left[
			n^{-1} \|\bar{\fZ}\|^2 \right]
			\leq
			\frac{15\mu^6}{n\lambda^2}
			\left(1+\frac{\mu}{\lambda}\right).
		\end{align}

		We next bound the integral over \( \mathcal{G}_{i_1} \) in \eqref{Eq.Pe1boundConv0Fast}. Consequently, under transmission of \(\fc_{i_1}\), where \(\mathbf Y(i_1)=\fc_{i_1}+\mathbf Z\),
		\begin{align}
			\int_{\mathcal{G}_{i_1}^c}
			f_{\fZ}(\fy-\fc_{i_1})\,d\fy
			=
			\Pr\{\mathbf Y(i_1)\notin\mathcal G_n(i_1)\}
			=
			\Pr\left\{\min_{1\leq t\leq n}Z_t<\varepsilon_n\right\} \stackrel{(a)}{\leq} \frac{e^{\lambda/\mu}}{\sqrt{\pi}}
			\sqrt{\frac{2}{\lambda}} n^{-\eta} \leq \gamma_2 .
			\label{Eq.GoodEventComplement}
		\end{align}
		where $(a)$ use \eqref{Ineq.Pr_Z_min} given in Lemma~\ref{Lem.AS_LB}. Thereby, combining \eqref{Eq.Pe1boundConv0Fast}, \eqref{Ineq.F_i_1_i_2_compl}, and \eqref{Eq.GoodEventComplement}, we obtain
		\begin{align}
			\label{Eq.ComplTypeIFast}
			1 - \gamma_1 - \gamma_2 - P_{e,1}(i_1)
			\leq
			\int_{\mathbbmss H_{i_1,i_2}}
			f_{\fZ}(\fy-\fc_{i_1})\,d\fy.
		\end{align}
		We now compare the two output densities on \(\mathbbmss H_{i_1,i_2}\). By construction, $\mathbbmss H_{i_1,i_2}\subseteq\mathcal G_n(i_1).$ Moreover, \((\fc_{i_1},\fc_{i_2})\in\mathcal P_n\). 	Observe that the absolute value of difference of noise distribution corresponding to distinct codewords $\fc_{i_1}$ and $\fc_{i_2}$ can be written as follows
		\begin{align}
			\label{Ineq.Error_Diff}
			\big| f_{\fZ}(\fy - \fc_{i_1}) - f_{\fZ}(\fy - \fc_{i_2}) \big|
			= f_{\fZ}(\fy - \fc_{i_1}) \cdot \bigg| 1 - \frac{f_{\fZ}(\fy - \fc_{i_2})}{f_{\fZ}(\fy - \fc_{i_1})} \bigg| .
		\end{align}
		By Lemma~\ref{Lem.Log_Likelihood_Ratio}, for any fixed \(\gamma_3>0\), the density ratio satisfies
		\begin{align}
			\left|
			\frac{f_{\fZ}(\fy-\fc_{i_2})}{f_{\fZ}(\fy-\fc_{i_1})}
			-1
			\right|
			\leq \gamma_3
		\end{align}
		for all sufficiently large \(n\), uniformly over all admissible codeword pairs. Thereby, we obtain
		\begin{align}
			\label{Ineq.GaussianContinuityFast}
			\big| f_{\fZ}(\fy - \fc_{i_1}) - f_{\fZ}(\fy - \fc_{i_2}) \big|
			=  f_{\fZ}(\fy - \fc_{i_1}) \cdot \bigg| 1 - \frac{f_{\fZ}(\fy - \fc_{i_2})}{f_{\fZ}(\fy - \fc_{i_1})} \bigg| \leq \gamma_3 f_{\fZ}(\fy - \fc_{i_1}),
		\end{align}
		 Next, integrating the pointwise density bound over \(\mathbbmss H_{i_1,i_2}\), we obtain
		 \begin{align}
		 	\left|
		 	\int_{\mathbbmss H_{i_1,i_2}}
		 	f_{\fZ}(\fy-\fc_{i_1})\,d\fy
		 	-
		 	\int_{\mathbbmss H_{i_1,i_2}}
		 	f_{\fZ}(\fy-\fc_{i_2})\,d\fy
		 	\right|
		 	&\quad\stackrel{(a)}{=}
		 	\left|
		 	\int_{\mathbbmss H_{i_1,i_2}}
		 	\left[
		 	f_{\fZ}(\fy-\fc_{i_1})
		 	-
		 	f_{\fZ}(\fy-\fc_{i_2})
		 	\right]\,d\fy
		 	\right|
		 	\nonumber\\
		 	&\quad\stackrel{(b)}{\leq}
		 	\int_{\mathbbmss H_{i_1,i_2}}
		 	\left|
		 	f_{\fZ}(\fy-\fc_{i_1})
		 	-
		 	f_{\fZ}(\fy-\fc_{i_2})
		 	\right|\,d\fy
		 	\nonumber\\
		 	&\quad\stackrel{(c)}{\leq}
		 	\gamma_3
		 	\int_{\mathbbmss H_{i_1,i_2}}
		 	f_{\fZ}(\fy-\fc_{i_1})\,d\fy
		 	\nonumber\\
		 	&\quad\stackrel{(d)}{\leq}\gamma_3,
		 	\label{Eq.IntegratedDensityDifference}
		 \end{align}
		 where
		 \begin{itemize}
		 	\item[(a)] follows from the linearity of integration;
		 	\item[(b)] follows from the triangle inequality for integrals, namely,
		 	\(\left|\int_A g(\fy)\,d\fy\right|
		 	\leq \int_A |g(\fy)|\,d\fy\);
		 	\item[(c)] follows from the pointwise density bound in \eqref{Ineq.GaussianContinuityFast}.
		 	\item[(d)] follows since
		 	\begin{align}
		 		\int_{\mathbbmss H_{i_1,i_2}}
		 		f_{\fZ}(\fy-\fc_{i_1})\,d\fy
		 		=
		 		\Pr\!\left(
		 		\fY(i_1) \in \mathbbmss H_{i_1,i_2}
		 		\right)
		 		\leq 1,
		 	\end{align}
		 	that is, the probability that the channel output belongs to any measurable subset of \(\mathbb R^n\) cannot exceed one.
		 \end{itemize}

		Since \(\mathbbmss H_{i_1,i_2}\subseteq\mathbbmss D_{i_1}\), we have
		\begin{align}
			P_{e,2}(i_2,i_1)
			=
			\int_{\mathbbmss D_{i_1}}
			f_{\fZ}(\fy-\fc_{i_2})\,d\fy
			\geq
			\int_{\mathbbmss H_{i_1,i_2}}
			f_{\fZ}(\fy-\fc_{i_2})\,d\fy.
			\label{Eq.TypeIIRestrictedRegion}
		\end{align}
		Combining these inequalities with
		\[
		1-\gamma_1-\gamma_2-P_{e,1}(i_1)
		\leq
		\int_{\mathbbmss H_{i_1,i_2}}
		f_{\fZ}(\fy-\fc_{i_1})\,d\fy,
		\]
		we obtain
		\begin{align}
			P_{e,1}(i_1)+P_{e,2}(i_2,i_1)
			\geq 1-\gamma_1-\gamma_2-\gamma_3.
			\label{Eq.FinalErrorLowerBound}
		\end{align}
		By the definition of the maximum error probabilities,
		\begin{align}
			\bar P_{e,1}^{(n)}+\bar P_{e,2}^{(n)}
			&\geq
			P_{e,1}^{(n)}(i_1)+P_{e,2}^{(n)}(i_2,i_1)
			\nonumber\\
			&\geq 1-\gamma_1-\gamma_2-\gamma_3.
		\end{align}
		Since \(\gamma_1,\gamma_2,\gamma_3>0\) can be chosen arbitrarily small, it follows that
		\[
		\liminf_{n\to\infty}
		\left[
		P_{e,1}(i_1)+P_{e,2}(i_2,i_1)
		\right]\geq 1.
		\]
		which contradicts the assumed vanishing-error criterion
		\[
		\bar P_{e,1}^{(n)}\to0,
		\qquad
		\bar P_{e,2}^{(n)}\to0.
		\]
		Therefore, the assumption in \eqref{Eq.alpha_n} is false, completing the proof of Lemma~\ref{Lem.Converse}.
		
	\end{proof}

	\section*{}
	\bibliographystyle{IEEEtran-full}
	\bibliography{Lit}
	
\end{document}

%% file: PKG.tex
\usepackage{etoolbox}
\usepackage{enumitem}
\usepackage[bottom]{footmisc}
\usepackage[dvipsnames,svgnames]{xcolor}
\usepackage[noadjust]{cite}
\usepackage{tikz}
\usetikzlibrary{arrows,arrows.meta,shapes.geometric,math}
\usepackage{pgfplots}
\pgfplotsset{compat=1.15}
\usepackage{mathtools}
\usepackage{nccmath}
\usepackage{amssymb}
\usepackage{mathdots}
\usepackage{upgreek}
\usepackage[caption=false]{subfig}
\usepackage{commath}
\usepackage{amsthm}
\usepackage{graphicx}
\usepackage{pgfgantt}
\usepackage{xfrac}
\usepackage{mathrsfs}
\usepackage{bbm}
\usepackage{multicol, blindtext}
\usepackage{amsbsy}
\usepackage{bm}
\usepackage{fixmath}
\usepackage{silence}
\usepackage{ctable}
\usepackage{mathdots}
\usetikzlibrary{math}

\usepackage{mathtools}
\def\nudge{.5}
\tikzset{axis/.style={ultra thick, Red!75!black, -latex, shorten <=-\nudge cm, shorten >=-2*\nudge cm}}
\tikzset{line/.style={thick,Green}}

\usepackage{tikz}
\usetikzlibrary{decorations.pathmorphing,calc}
\pgfmathsetseed{1} 
\pgfdeclarelayer{background}
\pgfsetlayers{background,main}

\usepackage[draft=false]{hyperref}

\usepackage{cleveref}
\crefname{equation}{Eq}{} 
\crefrangelabelformat{equation}{(#3#1#4--#5#2#6)}

\def\fc{{\bf c}}

\def\fx{{\bf x}}

\usepackage{autobreak}
\allowdisplaybreaks

\def \D{\mathcal{D}}
\def \E{\mathcal{E}}

\def \S{\mathcal{S}}

\def \fc{\mathbf{c}}
\def \fx{\mathbf{x}}

\def \bIG {\mathfrak{I}}

\def \fx{\mathbf{x}}
\def \fy{\mathbf{y}}
\def \fz{\mathbf{z}}

\def \fX{\mathbf{X}}
\def \fY{\mathbf{Y}}
\def \fZ{\mathbf{Z}}
\def \f0{\mathbf{0}}

\def\fc{{\bf c}}

\def\fx{{\bf x}}

\definecolor{blau_1a}{RGB}{93,133,195}
\definecolor{blau_2a}{RGB}{0,156,218}
\definecolor{gruen_3a}{RGB}{80,182,149}
\definecolor{gruen_4a}{RGB}{175,204,80}
\definecolor{gruen_5a}{RGB}{221,223,72}
\definecolor{orange_6a}{RGB}{255,224,92}
\definecolor{orange_7a}{RGB}{248,186,60}
\definecolor{rot_8a}{RGB}{238,122,52}
\definecolor{rot_9a}{RGB}{233,80,62}
\definecolor{lila_10a}{RGB}{201,48,142}
\definecolor{lila_11a}{RGB}{128,69,151}

\definecolor{blau_1b}{RGB}{0,90,169}
\definecolor{blau_2b}{RGB}{0,131,204}
\definecolor{gruen_3b}{RGB}{0,157,129}
\definecolor{gruen_4b}{RGB}{153,192,0}
\definecolor{gruen_5b}{RGB}{201,212,0}
\definecolor{orange_6b}{RGB}{253,202,0}
\definecolor{orange_7b}{RGB}{245,163,0}
\definecolor{rot_8b}{RGB}{236,101,0}
\definecolor{rot_9b}{RGB}{230,0,26}
\definecolor{lila_10b}{RGB}{166,0,132}
\definecolor{lila_11b}{RGB}{114,16,133}

\definecolor{mycolor1}{rgb}{0.0, 0.18, 0.39}
\definecolor{mycolor2}{RGB}{87,108,67}
\definecolor{mycolor3}{RGB}{8,133,161}
\definecolor{mycolor4}{RGB}{80,91,161}
\definecolor{mycolor5}{RGB}{98,122,157}
\definecolor{mycolor6}{RGB}{255,163,67}
\definecolor{mycolor7}{RGB}{152,205,225}
\definecolor{mycolor8}{RGB}{242,204,48}
\definecolor{mycolor9}{rgb}{0,.5,0}
\definecolor{mycolor10}{rgb}{.59,.44,.09}
\definecolor{mycolor11}{RGB}{231,199,31} 
\definecolor{mycolor12}{RGB}{8,133,161} 
\definecolor{mycolor13}{RGB}{157,188,64} 
\definecolor{mycolor14}{RGB}{194,150,130} 
\definecolor{mycolor15}{RGB}{98,122,157} 
\definecolor{mycolor16}{RGB}{160,160,160} 
\definecolor{mycolor17}{RGB}{115,82,68} 
\definecolor{mycolor18}{RGB}{94,60,108} 
\definecolor{mycolor19}{RGB}{115,82,68} 
\definecolor{mycolor20}{RGB}{255,183,30} 
\definecolor{mycolor21}{RGB}{13, 152, 186}
\definecolor{mycolor22}{RGB}{255, 223, 0}

\definecolor{cambridgeblue}{rgb}{0.64, 0.76, 0.68}
\definecolor{darkturquoise}{rgb}{0.0, 0.81, 0.82}
\definecolor{turquoiseblue}{rgb}{0.0, 1.0, 0.94}
\definecolor{turquoisegreen}{rgb}{0.63, 0.84, 0.71}
\definecolor{turquoise}{rgb}{0.19, 0.84, 0.78}
\definecolor{verdigris}{rgb}{0.26, 0.7, 0.68}
\definecolor{wheat}{rgb}{0.96, 0.87, 0.7}

\theoremstyle{remark} \newtheorem{theorem}{Theorem}
\theoremstyle{remark} \newtheorem{proposition}{Proposition}
\theoremstyle{remark} \newtheorem{definition}{Definition}
\theoremstyle{remark} 
\theoremstyle{remark} 
\theoremstyle{remark}

\newtheorem{innercustomgeneric}{\customgenericname}
\providecommand{\customgenericname}{}
\newcommand{\newcustomtheorem}[2]{%
	\newenvironment{#1}[1]
	{%
		\renewcommand\customgenericname{#2}%
		\renewcommand\theinnercustomgeneric{##1}%
		\innercustomgeneric
	}
	{\endinnercustomgeneric}
}
\newcustomtheorem{customcorollary}{Corollary}
\newcustomtheorem{customlemma}{Lemma}
\newcustomtheorem{customproposition}{Proposition}

%% file: Pos_PDf.tex
\begin{tikzpicture}[scale=1, transform shape]
	\begin{axis}[
		xmin=-20,
		xmax=80,
		ymin=0,
		ymax=.17,
		xlabel={$x~(\mu m)$},
		ylabel={$f_X(x;t)$},
		scaled y ticks=false,
		yticklabel style={
			/pgf/number format/fixed,
			/pgf/number format/precision=2
		},
		legend pos=north east,
		domain=-20:80,
		samples=400,
		width=12.5cm,
		height=7cm,
		title={Probability Density Function of Particle's Position},
		minor grid style={gray!20},
		major grid style={gray!50},
		legend style={
			at={(0.7,0.98)},
			anchor=north west,
			draw=none,
			fill=orange_6b!10,
			font=\small
		},
		scaled ticks=false,
		xmajorgrids,
		ymajorgrids,
		grid style={line width=0.2pt, draw=gray!40},
		minor tick num=4,
		]
		
		
		\addplot[very thick, color=blau_2b]
		{1/sqrt(2*pi*(6)) * exp(-((x - 1)^2)/(2*(6)))};
		\addlegendentry{$t=0.1\,s$}

		\addplot[very thick, color=blau_2b!80]
		{1/sqrt(2*pi*(24)) * exp(-((x - 10)^2)/(2*(24)))};
		\addlegendentry{$t=1\,s$}

		\addplot[very thick, color=blau_2b!50]
		{1/sqrt(2*pi*(64)) * exp(-((x - 30)^2)/(2*(64)))};
		\addlegendentry{$t=3\,s$}

		\addplot[very thick, color=blau_2b!30]
		{1/sqrt(2*pi*(104)) * exp(-((x - 50)^2)/(2*(104)))};
		\addlegendentry{$t=5\,s$}
		
	\end{axis}
\end{tikzpicture}

%% file: Inv_Gau_pdf.tex
\begin{tikzpicture}

	\begin{axis}[
		xmin=0,
		xmax=4,
		ymin=0,
		ymax=1.1,
		domain=.01:4,
		samples=300,
		width=12.5cm,
		height=7cm,
		xlabel={$z$},
		ylabel={$f_Z(z \mid \mu,\lambda)$},
		title={Parameter Effects on The IG PDF},
		grid=both,
		minor grid style={gray!20},
		major grid style={gray!50},
		legend style={
			at={(0.7,0.98)},
			anchor=north west,
			draw=none,
			fill=orange_6b!10,
			font=\small
		},
		legend pos=north east,
		tick label style={
			/pgf/number format/use comma=false,
			/pgf/number format/fixed,
			/pgf/number format/precision=2
		},
		xtick={0,1,2,3,4},
		ytick={0,0.5,1,1.5,2},
		scaled ticks=false,
		xmajorgrids,
		ymajorgrids,
		grid style={line width=0.2pt, draw=gray!40},
		minor tick num=4
		]
		
		
		\addplot[very thick, color=Green!40]
		{sqrt(10/(2*pi*x^3)) * exp(-(10*(x-2)^2)/(8*x))};
		\addlegendentry{$\mu=2,\ \lambda=10$}
		
		
	   \addplot[very thick, color=Green!90]
		{sqrt(16/(2*pi*x^3)) * exp(-(16*(x-2)^2)/(8*x))};
		\addlegendentry{$\mu=2,\ \lambda=16$}
		
		
		\addplot[very thick, color=blau_2b!40]
		{sqrt(1/(2*pi*x^3)) * exp(-(1*(x-1)^2)/(2*x))};
		\addlegendentry{$\mu=1,\ \lambda=1$}
		
		
		\addplot[very thick, color=blau_2b!90]
		{sqrt(1/(2*pi*x^3)) * exp(-(1*(x-6)^2)/(72*x))};
		\addlegendentry{$\mu=6,\ \lambda=1$}
		
	\end{axis}
\end{tikzpicture}

%% file: Sys_Mdl.tex
\begin{tikzpicture}[xscale=70, yscale=70]
	
	\def\d{0.15}            		
	\def\steps{100}				
	\def\sigma{0.002}    	
	\def\v{0.001}			

	\draw[->] (-0.01,0) -- (\d+0.01,0) node[right] {$x$};
	\draw[->] (0,-0.04) -- (0,0.04) node[above] {$y$};
	
	\draw[thick, dashed] (\d,-0.04) -- (\d,0.04);
	\node[right] at (\d,0.04) {Rx};
	
	\fill[blau_2b] (0,0) circle (0.0014);
	\node[below left] at (0,0) {Tx};
	
	\pgfmathsetmacro{\x}{0}
	\pgfmathsetmacro{\y}{0}
	
	\foreach \i in {1,...,\steps} {
		
		\pgfmathsetmacro{\u}{max(rnd,0.0001)}
		\pgfmathsetmacro{\w}{rnd}
		
		\pgfmathsetmacro{\dx}{\sigma*sqrt(-2*ln(\u))*cos(360*\w)}
		\pgfmathsetmacro{\dy}{\sigma*sqrt(-2*ln(\u))*sin(360*\w)}
		
		\pgfmathsetmacro{\xnew}{\x + \v + \dx}
		\pgfmathsetmacro{\ynew}{\y + \dy}
		
		\draw[SkyBlue, thick] (\x,\y) -- (\xnew,\ynew);
		
		\xdef\x{\xnew}
		\xdef\y{\ynew}
	}
	
	\fill[Dandelion] (\x,\y) circle (0.001);
	
	\draw[<->, thin] (0,-0.03) -- (\d,-0.03) node[midway, below] {$d=0.15mm$};
	
\end{tikzpicture}